\documentclass[a4paper]{article}

\usepackage{amsmath,amsthm,amssymb,latexsym,indentfirst}
\RequirePackage{natbib}
\usepackage[a4paper,left=3cm, right=3cm, top=2cm,bottom=3cm,nomarginpar, headheight=23pt,includeheadfoot]{geometry}
\newtheorem{theorem}{Theorem}[section]
\newtheorem{proposition}[theorem]{Proposition}
\newtheorem{lemma}[theorem]{Lemma}
\newtheorem{corollary}[theorem]{Corollary}
\newtheorem{remark}[theorem]{Remark}
\newtheorem{definition}[theorem]{Definition}
\newtheorem{example}[theorem]{Example}
\newtheorem{assumption}[theorem]{Assumption}

\def \Int{\displaystyle\int}

\usepackage{xcolor}
\begin{document}

\title{On Optimal Investment for
a Behavioural Investor in Multiperiod Incomplete Market Models}
\author{Laurence Carassus\\ LPMA, Universit\'e Paris Diderot-Paris 7\\ LMR, URCA \and Mikl\'os R\'asonyi\\ University of Edinburgh
\thanks{M. R\'asonyi thanks
University Paris Diderot-Paris 7
for an invitation in 2010 during which part of this research was carried out and  he dedicates this paper to A. Brecz.}}

\date{\today}

\maketitle

\begin{abstract}
We study the optimal investment problem for a behavioral investor in
an incomplete discrete-time multiperiod financial market model. For
the first time in the literature, we provide easily verifiable and
interpretable conditions for well-posedness. Under
two different sets of assumptions we also establish the
existence of optimal strategies.
\end{abstract}
{\bf Keyword :} Optimisation, existence and well-posedness in
behavioral finance, ``$S$-shaped'' utility function, probability distortion,
Choquet integral.

\section{Introduction}
A classical optimization problem of mathematical finance is to find
the investment strategy that maximizes the expected von
Neumann-Morgenstern utility (\cite{vnm44}) of the portfolio value of some economic
agent, see e.g. Chapter 2 of \cite{fs}. In mathematical terms,
$Eu(X)$ needs to be maximized in $X$ where $u$ is a concave
increasing function and $X$ runs over possible values of admissible
portfolios. Note that the concavity of $u$ refers to the risk aversion of the economic agent.
Since 1947, this approach has been intensively used to model investor behaviour towards risk.
However, as shown by \cite{a53}, one of the fundamental axioms of the von
Neumann-Morgenstern theory is often violated empirically from the observed behaviour of agents.

{Based on experimentation, \cite{kt} introduced the cumulative
prospect theory, which provided a possible solution for the Allais
paradox.} First, this theory asserts that the problem's mental
representation is important: agents analyze their gains or losses
with respect to a given stochastic reference point $B$ rather than
to zero. Second, \cite{kt} assert that potential losses are taken
into account more than potential gains. So agents behave differently
on gains, i.e. on $(X-B)_+$ (where $X$, again, runs over possible
values of admissible portfolios) and on losses, i.e. on $-(X-B)_-$.
Third, agents overweight events with small probabilities (like
extreme events) and underweight the ones with large probabilities.
This can be translated into mathematics by the following
assumptions: investors use an ``$S$-shaped'' utility function $u$
(i.e. $u(x)=u_+(x)$, $x\geq 0$; $u(x)=-u_-(-x)$, $x<0$ where
$u_+,u_-:\mathbb{R}_+\to\mathbb{R}$ are concave and increasing.
{\cite{kt} assume also that $u_-$ is ``stronger" than
$u_+$: $u_-=2,25 u_+$. Next, the investors} distort the
probability measure by a transformation function of the cumulative
distributions: instead of expectations, Choquet integrals appear.
Furthermore, maximization of their objective function takes place
over the random variables of the form $X-B$.

That paper triggered an avalanche of subsequent investigations,
especially in the economics literature, see e.g. the references of \cite{jz} and \cite{cd}. But from the mathematical side the first
significant step ahead is due, quite recently, to \cite{jz}.
% for a sample relevant to our present discussions.
This late development, as pointed out in \cite{jz}, is explained by the presence of massively difficult
obstacles: the objective function is non-concave and the probability
distortions make it impossible to use dynamic programming and the related machinery
based on the Bellmann equation.

Up to now two types of models have been studied: complete continuous-time models or one-step models. Here, for the first time
in the literature (to the best of our knowledge) we propose
results in incomplete multiperiod discrete time models.

%In \cite{jz} a rather specific continuous-time market model (driven by Brownian motion) is treated. As this
%model is complete (all ``reasonable'' random variables can be realized
%by continuous trading), the behavioral investment problem can be reduced to a (still very
%difficult) static optimization over a set of random variables.

The existing studies in continuous time models heavily rely on completeness of the market (i.e. all ``reasonable'' random variables can be realized
by continuous trading): see for example \cite{jz} or \cite{cd}. They also make assumptions on the portfolio losses. \cite{cd} allow only portfolios whose
attainable wealth is bounded from below by $0$; in \cite{jz}  the portfolio
may admit losses, but this loss must be bounded from below by a constant
(which may depend, however, on the chosen strategy). Recall, however, that when the (concave) utility
function $u$ is defined on the whole real line, standard utility
maximisation problems usually admit optimal solutions that are \emph{not} bounded
from below, see \cite{sch}.
Note also the papers of \cite{p08} and \cite{cd11} which proposed explicit evaluations of the optimal solution for some specific utility functions.

It thus seems desirable to investigate models which are incomplete
and which allow portfolio losses that can be unbounded from below.
In this paper, we focus on discrete time models, which are
generically incomplete. In \cite{bg} and \cite{hz}, a single period
model is studied. This is the first
mathematical treatment of discrete-time multiperiod incomplete
models in the literature. We allow for a possibly stochastic
reference point $B$. More interestingly, we need no concavity or
even monotonicity assumptions on $u_+,u_-$: only their behavior at
infinity matters. Note that in \cite{jz} and \cite{cd} the functions
$u_+,u_-$ are assumed to be concave and the reference point is
easily incorporated: as the market is complete any stochastic
reference point can be replicated. This is no longer so in our
incomplete setting.

The issue of well-posedness is a recurrent theme in related papers
(see \cite{bg}, \cite{hz}, \cite{jz} and \cite{cd11}). To the best of our knowledge, our Theorem \ref{ajtothm}
below is the first positive result on well-posedness for discrete-time multiperiod
models.
%The only continuous-time reference we are aware of is \cite{jz}.

%The conditions
%for well-posedness
%in \cite{jz} are not easily checked and have no
%obvious interpretations. Here
In Theorem \ref{ajtothm} we manage to provide intuitive and easily
verifiable conditions which apply {to a broad class of
functions $u_+,u_-$ and of probability distortions (see Assumption
\ref{as1} and Remark \ref{rmqas1})
%in the case where $u_+,u_-$ and the probability distortions are all
%``power-like'' functions (as proposed in \cite{tk}) satisfying
%certain parameter constraints
as soon as appropriate moment
conditions hold for the price process. We also provide examples
highlighting the kind of parameter restrictions which are necessary
for well-posedness in a multiperiod context. It turns out that
multiple trading periods exhibit phenomena which are absent in the
one-step case.
%Related results for continuous-time markets will be presented in the
%companion paper \cite{mr}.
%We compare them to our sufficient conditions.

Existence of optimal strategies is fairly subtle in this setting as no dynamic
programming is possible and there is a lack of concavity and hence popular
compactness substitutes (such as the Koml\'os theorem) do not apply. More surprisingly, and in contrast to the usual maximization of expected utility, it turns out that the investor may increase her satisfaction by exploiting randomized trading strategies.
We provide two types of existence result. The first one (see Theorem \ref{maine} below) use ``relaxed'' strategies: we assume that the strategies are measurable with
respect to some information flow, which have a certain structure and, in particular, allow the use of an external source
of randomness (see Assumption \ref{as3}).  The second existence result (see Theorem
\ref{maine-new} below) is proved for ``pure'' strategies if the information filtration is
rich enough (see Assumption \ref{AAA} which is satisfied by classical incomplete models): there is no need for an external random source.

The standard (concave) utility maximisation machinery provides powerful
tools for risk management as well as for pricing in incomplete markets.
We hope that our present results are not only of theoretical interest but also contribute to the development of a similarly applicable framework for investors with behavioural criteria.

%The optimal strategy is constructed
%using an additional source of randomness. We give evidence in Remark \ref{mahler}
%why this is reasonable.
%Finally, we exhibit examples showing that
%our assumptions are satisfied in a broad class of models.

%We also point out an unorthodox feature that is absent
%in the case of maximizing expected utility:
%randomizing strategies may improve performance.

The paper is organized as follows: in section \ref{ketto} we introduce
notation and the market model; section
\ref{harom} presents examples pertinent to the well-posedness of the problem; section
\ref{negy} provides a sufficient condition for well-posedness in a multiperiod market; section \ref{beszur} discusses a relaxation of the set of trading
strategies based on an external random source;
section \ref{ot} proves the existence of optimal portfolios under appropriate conditions
using ``relaxed'' controls which exploit an external random source;
section \ref{secondrevision} proves an existence result for the set
of ordinary controls provided that the information filtration is rich enough; section
\ref{exaoa} exhibits examples showing that our assumptions are satisfied in
a broad class of market models; finally, section \ref{appendix} contains most
of the proofs as well as some auxiliary results.

\section{Market model description}\label{ketto}

Let $(\Omega,\mathcal{F},(\mathcal{F}_t)_{0\leq t\leq T},P)$ be a
discrete-time filtered probability space with time horizon
$T\in\mathbb{N}$. We will often need the set of $m$-dimensional
$\mathcal{F}_t$-measurable random variables, so we introduce the notation $\Xi_t^m$ for this
set.

Let $\mathcal{W}$ denote the set of $\mathbb{R}$-valued random
variables $Y$ such that $E\vert Y\vert^p <\infty$ for all $p>0$.
This family is clearly closed under addition, multiplication and
taking conditional expectation. The family of nonnegative elements in $\mathcal{W}$
is denoted by $\mathcal{W}^+$. With a slight abuse of
notation, for a $d$-dimensional random variable $Y$, we write
$Y\in\mathcal{W}$ when we indeed mean $\vert Y\vert\in\mathcal{W}$.
We will also need $\mathcal{W}^+_t:=\mathcal{W}^+\cap\Xi_t^1$.

When defining objects using an equality we will use the symbol $:=$
in the sequel. Let $\gamma>0$, $X$ be some random variable and
$A\in\mathcal{F}$ an event. We will use the following notations:
$$
P^{\gamma}(A|\mathcal{F}_{t}):=(P(A|\mathcal{F}_{t}))^{\gamma}\quad
E^{\gamma}(X|\mathcal{F}_{t}):=(E(X|\mathcal{F}_{t}))^{\gamma}.
$$

Let $\{S_t,\ 0\leq t\leq T\}$ be a $d$-dimensional adapted process
representing the (discounted) price
of $d$ securities in the financial market in consideration. The notation $\Delta
S_t:=S_t-S_{t-1}$ will often be used. Trading strategies are given
by $d$-dimensional processes $\{\theta_t,\ 1\leq t\leq T\}$ which
are supposed to be predictable (i.e. {$\theta_t \in \Xi_{t-1}^d$})
%is $\mathcal{F}_{t-1}$-measurable).
The class of all such strategies
is denoted by $\Phi$.

%The notation := will stand for definition. For all $A \in \mathcal{F}_{t}$, $X \in \Xi_{t+1}^m$,
%$\gamma \in \mathbb{R}$, we will use the following notations

Trading is assumed to be self-financing, so the value  of a
portfolio strategy $\theta\in\Phi$ at time $0\leq t\leq T$ is
$$
X_t^{X_0,\theta}:=X_0+\sum_{j=1}^t \theta_j\Delta S_j,
$$
where $X_0$ is the initial capital of the agent in consideration and
the concatenation $xy$ of elements $x,y\in \mathbb{R}^d$ means
that we take their scalar product.

Consider the following technical condition (R). It says, roughly speaking, that there are
no redundant assets, even conditionally, see also Remark 9.1 of \cite{fs}.

\medskip

(R) {\em The support of the (regular)
conditional distribution of $\Delta S_t$ with respect to
$\mathcal{F}_{t-1}$ is not contained in any proper affine subspace
of $\mathbb{R}^d$, almost surely, for all $1\leq t\leq T$.}

\begin{remark}
\label{remdt} {\rm Dropping (R) and modifying Assumption \ref{marche}
in an appropriate way proofs go through but they get very messy.
In this case one should consider suitably defined projections of the strategies on
the affine hull figuring in condition (R). }
\end{remark}

\medskip

The following absence of arbitrage condition is standard, it is
equivalent to the existence of a risk-neutral measure in discrete
time markets with finite horizon, see e.g. \cite{dmw}.

\medskip

(NA) {\em If $X^{0,\theta}_T\geq 0$ a.s. for some $\theta\in\Phi$ then
$X^{0,\theta}_T=0$ a.s.}

\medskip

The next proposition is a trivial reformulation of Proposition 1.1
in \cite{cr}.

\begin{proposition}\label{bibi}  The condition {\rm (R) + (NA)} is equivalent to the existence
of $\mathcal{F}_t$-measurable random variables
$\kappa_t,\pi_t>0$, $0\leq t\leq T-1$ such that
\begin{equation*}\label{mi}
\mathrm{ess.}\inf_{\xi\in{\Xi}_t^d} P(\xi\Delta
S_{t+1}\leq -\kappa_t\vert\xi\vert\,| \mathcal{F}_t)\geq\pi_t\mbox{ a.s.}
\end{equation*}
%\hfill$\Box$
\end{proposition}
%\begin{proof} The direction $(R)+(NA)\Rightarrow\mbox{\eqref{mi}}$ follows from
%Proposition 3.3 of \cite{stettner}. The
%other direction is clear from the implication $(g)\Rightarrow (a)$
%of Theorem 3 in \cite{jacod} and from the fact
%that if $(R)$ failed we would have, for some
%$\xi\in\tilde{\Xi}_t$,
%$$
%P(\langle\xi,\Delta S_{t+1}\rangle=0\vert\mathcal{F}_t)=1
%$$
%on a set of positive measure, contradicting \eqref{mi}.
%\end{proof}

We now present the hypotheses on the market model that will be needed
for our main results in the sequel.

\begin{assumption}\label{marche}
For all $t\geq 1$, $\Delta S_t\in\mathcal{W}_t$. Furthermore,
for $0\leq t\leq T-1$, there exist
$\kappa_t,\pi_t>0$ satisfying $1/\kappa_t,1/\pi_t\in\mathcal{W}^+_t$ such that
\begin{eqnarray}\label{rrr}
\mathrm{ess.}\inf_{\xi\in\Xi_{t}^d}
P( \xi \Delta S_{t+1}\leq -\kappa_{t} \vert\xi\vert  | \mathcal{F}_{t})
\geq \pi_{t} \mbox{ a.s.}
\end{eqnarray}
%\end{enumerate}
\end{assumption}

The first item in the above assumption could
be weakened to the existence of the $N$th moment for $N$ large
enough but this would lead to complicated book-keeping with no essential
gain in generality, which we
prefer to avoid. In the light of Proposition \ref{bibi}, \eqref{rrr}
is a certain strong form of no-arbitrage. Note that if either $\kappa_t$ or $\pi_t$ is not constant, then
even a simple von Neumann-Morgenstern utility maximisation problem may be ill posed (see Example 3.3 in \cite{cr}).

Section \ref{exaoa} below
exhibits concrete examples showing that
Assumption \ref{marche} holds in a broad class of market models.
We note that,
by Proposition \ref{bibi}, Assumption \ref{marche} implies
both (NA) and (R) above.

Now we turn to investors' behavior, as modeled by cumulative prospect theory,
see \cite{kt,tk}. Agents' attitude towards gains and losses will be
expressed by the functions $u_+$ and $u_-$. Agents are assumed to have a
(possibly stochastic) reference point $B$ and probability distortion functions $w_+$ and $w_-$.

Formally, we assume that
$u_{\pm}:\mathbb{R}_+\to\mathbb{R}_+$ and $w_{\pm}: [0,1]\to [0,1]$  are
measurable functions such that
$u_{\pm}(0)=0$, $w_{\pm}(0)=0$ and $w_{\pm}(1)=1$.
We fix $B$, a scalar-valued random variable in $\Xi_T^1$.

%\begin{remark}
%{\rm I think continuity and monotonicity of $u_{\pm}$ and $w_{\pm}$ are not useful (I have changed
%the minoration of $V^-$ in Lemma \ref{eel}).}
%\end{remark}

%Formally, we assume that
%$u_+,u_-:\mathbb{R}_+\to\mathbb{R}_+$ are nondecreasing, continuous functions
%with $u_+(0)=u_-(0)=0$.
%We fix $B$, a scalar-valued random variable.
%The functions $w_+,w_-: [0,1]\to [0,1]$ are nondecreasing and continuous with
%$w_+(0)=w_-(0)=0$ and $w_+(1)=w_-(1)=1$.
%
%\begin{remark}
%{\rm Actually, the main results of the present article need less
%about $u_{\pm},w_{\pm}$ than stipulated here, in particular, about
%continuity and monotonicity. But as these are fairly natural
%requirements for agents' preferences, we make these assumptions
%throughout the paper.}
%\end{remark}

\begin{example}\label{extraa}
{\rm A typical choice is taking
$$
u_+(x)=x^{\alpha_+},\quad u_-(x)=kx^{\alpha_-}
$$
for
some {$k>0$} and setting
$$
w_+(p)=\frac{p^{\gamma_+}}{(p^{\gamma_+}+(1-p)^{\gamma_+})^{1/\gamma_+}},\quad
w_-(p)=\frac{p^{\gamma_-}}{(p^{\gamma_-}+(1-p)^{\gamma_-})^{1/\gamma_-}},
$$
with constants {$0<\alpha_{\pm},\gamma_{\pm}\leq 1$}. In \cite{tk}, based on experimentation,
the following choice was made: $\alpha_{\pm}=0.88$, $k=2.25$,
$\gamma_+=0.61$ and $\gamma_-=0.69$.}
\end{example}

We define, for $X_0\in \Xi_0^1$ and $\theta \in \Phi$,
\begin{eqnarray*}
V^+(X_0;\theta_1,\ldots,\theta_{T}):=\Int_0^{\infty} w_+\left(P\left(
u_+\left(\left[X^{X_0,\theta}_T-B\right]_+\right)\geq y
\right)\right)dy,
\end{eqnarray*}
and
\begin{eqnarray*}
V^-(X_0;\theta_1,\ldots,\theta_{T}):=\Int_0^{\infty} w_-\left(P\left(
u_-\left(\left[X^{X_0,\theta}_T-B\right]_-\right)\geq y
\right)\right)dy,
\end{eqnarray*}
and whenever $V^-(X_0;\theta_1,\ldots,\theta_{T})<\infty$ we set
\[
V(X_0;\theta_1,\ldots,\theta_{T}):=V^+(X_0;\theta_1,\ldots,\theta_{T})-V^-(X_0;\theta_1,\ldots,\theta_{T}).
\]

\noindent We denote by $\mathcal{A}(X_0)$ the set of strategies $\theta \in \Phi$ such that
\[
V^-(X_0;\theta_1,\ldots,\theta_{T})<\infty
\]
and we call them {\it admissible} (with respect to $X_0$).

\begin{remark}
{\rm If there were no probability distortions (i.e. $w_{\pm}(p)=p$) then
we would simply get}
$
V^+(X_0;\theta_1,\ldots,\theta_{T})=Eu_+\left(\left[X^{X_0,\theta}_T-B\right]_+\right)$ {\rm and} $
V^-(X_0;\theta_1,\ldots,\theta_{T})=Eu_-\left(\left[X^{X_0,\theta}_T-B\right]_-\right)$
{\rm and hence $V(X_0;\theta_1,\ldots,\theta_{T})$ equals the expected utility
$Eu(X^{X_0,\theta}_T-B)$ for the utility function $u(x)=u_+(x)$, $x\geq 0$,
$u(x)=-u_-(-x)$, $x<0$.}

{\rm We refer to \cite{cp} for the explicit treatment of this
problem in a continuous time, complete case under the assumptions
that $u_+$ is concave, $u_-$ is convex (hence $u$ is piecewise
concave) and $B$ is deterministic. In \cite{bkp} this problem is studied again
in a complete, continuous time model but for a power convex-convave shaped utility function.
In \cite{cr11} this problem is investigated in a general discrete-time multiperiod model
under the hypothesis that the (suitably defined) asymptotic elasticity of
$u_-$ is strictly greater than that of $u_+$. }
\end{remark}

The present paper is concerned with maximizing $V(X_0;\theta_1,\ldots,\theta_T)$
over $\theta\in\mathcal{A}(X_0)$.
%(see  \eqref{At} for the definition of $\mathcal{A}_0(X_0)$).
We seek to find conditions ensuring {\it well-posedness}, i.e.
\begin{equation}\label{plumduff}
\sup_{\theta\in\mathcal{A}(X_0)} V(X_0;\theta_1,\ldots,\theta_T)<\infty,
\end{equation}
and the existence of $\theta^*\in\mathcal{A}(X_0)$ attaining this supremum.

\begin{remark}\label{qtya}
{\rm One may wonder whether the set $\mathcal{A}(X_0)$ is rich
enough. Assume that $u_-(x)\leq c(1+x^{\eta})$ for some $c, \eta
>0$, $X_0,B\in\mathcal{W}$ and $w_-(p)\leq Cp^{{\delta_-}}$ for some
$0<{\delta_-} \leq 1$ and $C >0$. Then Lemma
\ref{magi} below implies that the strategy $\theta_t=0$,
$t=1,\ldots,T$ is in $\mathcal{A}(X_0)$, in particular, the latter
set is non-empty. If, furthermore, $\Delta S_t\in\mathcal{W}_t$ for
all $t$ then $\theta\in\mathcal{A}(X_0)$ whenever
$\theta_t\in\mathcal{W}_{t-1}$, $t=1,\ldots,T$. This remark applies,
in particular, to $u_-$ and $w_-$ in Example \ref{extraa} above.}
\end{remark}

\section{A first look at well-posedness}\label{harom}

In this section we find parameter restrictions that need to hold
in order to have a well-posed problem in the setting of e.g.
Example \ref{extraa}. The discussion below sheds light on
the assumptions we will make later in section \ref{negy}.

For simplicity we assume that $u_+(x)=x^{\alpha_+}$ and
$u_-(x)=x^{\alpha_-}$ for some $0<\alpha_{\pm}\leq 1$; the
distortion functions are $w_+(t)=t^{\gamma_+}$,
$w_-(t)=t^{\gamma_-}$ for some $0< \gamma_{\pm}\leq 1$. The example
given below applies also to $w_{\pm}$ with a power-like behavior
near $0$ such as those in Example \ref{extraa} above.

Let us consider a two-step market model with $S_0=0$, $\Delta S_1$
uniform on $[-1,1]$, $P(\Delta S_2=\pm 1)=1/2$ and $\Delta S_2$ is
independent of $\Delta S_1$. Let
$\mathcal{F}_0,\mathcal{F}_1,\mathcal{F}_2$ be the natural
filtration of $S_0,S_1,S_2$. It is easy to check that Assumption
\ref{marche} holds with $\kappa_0=\kappa_1=1/2$, $\pi_0=1/4$ and $\pi_1=1/2$.

Let us choose initial capital $X_0=0$ and reference point $B=0$.
%We consider $\mathcal{A}(X_0)$, the set of strategies $(\theta_1,\theta_2)$ where $\theta_i$ is $\mathcal{F}_{i-1}$-measurable $i=1,2$ such that $V_-(X_0;\theta_1,\theta_2)<\infty$.
We consider the strategy $\theta\in \Phi$
given by $\theta_1=0$ and $\theta_2=g(\Delta S_1)$ with
$g:[-1,1)\to[1,\infty)$ defined by $g(x)=(\frac2{1-x})^{1/\ell}$, where $\ell>0$ will be chosen later.
Then the distribution function of $\theta_2$ is given by
\[
F(y)=0,\ y<1,\quad F(y)=1-\frac{1}{y^{\ell}},\ y\geq 1.
\]
It follows
that
\begin{eqnarray*}\label{2}
V^+(0;\theta_1,\theta_2)=\int_0^{\infty}
P^{\gamma_+}((\theta_2\Delta S_2)_+^{\alpha_+}\geq y)dy=
\int_1^{\infty} \frac{1}{2^{\gamma_+}}\frac{1}{y^{\ell\gamma_+/\alpha_+}} dy,
\end{eqnarray*}
and
\begin{eqnarray*}\label{2a}
V^-(0;\theta_1,\theta_2)=\int_0^{\infty} P^{\gamma_-}((\theta_2\Delta S_2)_-^{\alpha_-}\geq y)dy=
\int_1^{\infty} \frac{1}{2^{\gamma_-}}\frac{1}{y^{\ell\gamma_-/\alpha_-}} dy.
\end{eqnarray*}

If we have ${\alpha_+}/{\gamma_+}>{\alpha_-}/{\gamma_-}$ then there
is $\ell>0$ such that
\begin{eqnarray*}\label{3a}
\frac{\ell \gamma_+}{\alpha_+}<1<\frac{\ell \gamma_-}{\alpha_-},
\end{eqnarray*}
which entails $V^-(0;\theta_1,\theta_2)<\infty$ (so indeed $\theta\in\mathcal{A}(0)$) and
$V^+(0;\theta_1,\theta_2)=\infty$ so the optimization problem becomes ill-posed.

One may wonder whether this phenomenon could be ruled out by restricting
the set of strategies e.g. to bounded ones. The answer is no. Considering
$\theta_1(n):=0, \theta_2(n):=\min\{\theta_2,n\}$ for $n\in\mathbb{N}$ we obtain easily that $\theta(n)\in\mathcal{A}(0)$ and
$V^+(0;\theta_1(n),\theta_2(n))\to \infty$, $V^-(0;\theta_1(n),\theta_2(n))\to
V^-(0;\theta_1,\theta_2)<\infty$ by monotone convergence, which shows
that we still have
$$
\sup_{\psi}V(0;\psi_1,\psi_2)=\infty,
$$
where $\psi$ ranges over the family of bounded strategies of $\mathcal{A}(0)$
%\cap\{\psi:\, \psi\mbox{ is bounded}\}$
only.
This shows that the ill-posedness phenomenon is not just a pathology
but comes from the multi-periodic setting: one may use the information available
at time $1$ when choosing the investment strategy $\theta_2$.

We mention another case of ill-posedness which is present already in
one-step models, as noticed in \cite{hz} and \cite{bg}. We
slightly change the previous setting.
We choose $u_+(x)=x^{\alpha_+}$ and $u_-(x)=kx^{\alpha_-}$, for $k>0$ and $0<\alpha_{\pm}\leq 1$.
We allow general distortions, assuming only that
$w_{\pm}(p)>0$ for $p>0$. The market is defined by $S_0=0$,
$\Delta S_1=\pm 1$ with probabilities $p,1-p$ for some $0<p<1$ and
$\mathcal{F}_0,\mathcal{F}_1$ the natural filtration of $S_0,S_1$.
Now the set $\mathcal{A}(X_0)$ can be identified with $\mathbb{R}$
(i.e. with the set of $\mathcal{F}_0$-measurable random
variables). Take
$X_0=B=0$ and $\theta_1(n):=n$, $n\in\mathbb{N}$, then
$V^+(0;\theta_1(n))= w_+(p)n^{\alpha_+}$ and
$V^-(0;\theta_1(n))= kw_-(1-p)n^{\alpha_-}$. If
$\alpha_+>\alpha_-$ then, \emph{whatever $w_+,w_-$ are}, we have
$V(0;\theta_1(n))\to\infty$, $n\to\infty$. Hence, in order to get a
well-posed problem one needs to have $\alpha_+\leq\alpha_-$, as
already observed in \cite{bg} and \cite{hz}.

We add a comment on the case $\alpha_+=\alpha_-$ assuming, in
addition, that $w_+,w_-$ are e.g. continuous~: \emph{whatever
$w_+,w_-$ are}, we may easily choose $p$ such that the problem
becomes ill-posed: indeed, it happens if  $w_+(p)>kw_-(1-p)$. This
shows, in particular, that even in such very simple market models
the problem with the parameter specifications of \cite{tk} can be
ill-posed (e.g. take any $p>0,788$ and consider the setting of
Example \ref{extraa} with the parameters of \cite{tk} quoted there).
We interpret this fact as follows: the participants of the
experiments conducted by \cite{tk} would perceive that such market
opportunities may lead to their arbitrary (inifinite) satisfaction.

Since it would
be difficult to dismiss the simple models of this section based on economic grounds we are led to the conclusion that, in order to get a mathematically
meaningful optimization problem for a reasonably wide range of price processes, one needs to assume both
\begin{equation}\label{bulb}
\alpha_+<\alpha_-\quad\mbox{ and }\quad \alpha_+/\gamma_+\leq\alpha_-/\gamma_-.
\end{equation}

In the following section we propose an easily verifiable sufficient
condition for the well-posedness of this problem in multiperiod
discrete-time market models. The decisive condition we require is
${\alpha_+}/{\gamma_+}<\alpha_-$, see \eqref{a} below. This is
stronger than \eqref{bulb} but still reasonably general. If
$w_-(p)=p$ (i.e. $\gamma_-=1$, no distortion on loss probabilities)
then \eqref{a} below is essentially sharp, as the present section
highlights.

\section{Well-posedness in the multiperiod case}\label{negy}

%Assumptions \ref{marche}, \ref{as1} and \ref{as2} will be in force throughout this section.
In this section,
after introducing the conditions we need on
$u_{\pm},w_{\pm}$, we will prove our sufficient condition for the well-posedness
of the behavioural investment problem (Theorem \ref{ajtothm}).

Basically, we require that $u_{\pm}$ behave in a power-like way at
infinity ({this is automatically true for any function having
bounded from above positive asymptotic elasticity and bounded from below negative asymptotic elasticity, see
Remark \ref{rmqas1}}) and $w_{\pm}$ do likewise in the neighborhood of
$0$. We stress that no concavity, continuity or monotonicity
assumptions are made on $u_{\pm}$, unlike in all related papers.

\begin{assumption}\label{as1} We assume that
$u_{\pm}:\mathbb{R}_+\to\mathbb{R}_+$ and $w_{\pm}: [0,1]\to [0,1]$  are
measurable functions such that
$u_{\pm}(0)=0$, $w_{\pm}(0)=0$ and $w_{\pm}(1)=1$ and
\begin{eqnarray}\label{tenet}
u_+(x) &\leq& k_+ (x^{\alpha_+} +1),\\
\label{tenet-}
k_- (x^{\alpha_-}-1) &\leq& u_-(x),\\
w_+(p)&\leq& g_+  p^{\gamma_+},\label{fuga}\\
w_-(p)&\geq& g_- p \label{b},
\end{eqnarray}
with $0<\alpha_{\pm},\gamma_+\leq 1$, $k_{\pm},g_{\pm} >0$ fixed
constants and
\begin{eqnarray}\label{a}
\frac{\alpha_+}{\gamma_+} < \alpha_-.
\end{eqnarray}
This allows us to fix $\lambda$ such that $\lambda\gamma_+>1$ and
$\lambda\alpha_+<\alpha_-$.
\end{assumption}

\begin{remark}
\label{rmqas1}
{\rm The condition $\alpha_{\pm},\gamma_+\leq 1$ is not necessary
for our results to hold true, it is just stated for ease of exposition.
We first comment on \eqref{tenet} and \eqref{tenet-}. Define the
utility function $u(x)=u_+(x)$, $x\geq 0$, $u(x)=-u_-(-x)$, $x<0$
and assume that $u_{\pm}$ are differentiable. Then $AE_+(u)  =
\limsup_{x\to\infty} \frac{u'(x)x}{u(x)} \leq {\alpha}_+$ implies
\eqref{tenet} and $AE_-(u) =  \liminf_{x\to
-\infty}\frac{u'(x)x}{u(x)}  \geq {\alpha}_-$ implies
\eqref{tenet-}. So only the behavior of $u_{\pm}$ near infinity
matters. This comes from Lemma 6.3 (i) of \cite{KS99} (their proof
does not rely on concavity) which asserts {the existence of some
$x_0>0$ such that for all $x\geq x_0$, $\rho \geq 1$, $u_+(\rho x)
\leq \rho ^{{\alpha}_+}u_+(x)$. So for $x\geq x_0$, choosing $\rho
=\frac{x}{x_0}$}, we get that $u_+(x) \leq x^{\alpha_+}
\frac{u_+(x_0) }{x_0^{\alpha_+}}$ and we conclude that \eqref{tenet}
holds since for $0<x\leq x_0$, $u_+(x) \leq u_+(x_0)$. The proof for
\eqref{tenet-} is similar. Note that if {$w_{\pm}(p)=p$}
%$\gamma_{\pm}=1$,
we prove
in \cite{cr11} an existence result under the condition \eqref{a},
which asserts in this case that $AE_-(u)>AE_+(u)$.

Condition \eqref{a} has already been mentioned
in the previous section. It has a rather straightforward
interpretation: the investor takes losses more seriously than gains.
The distortion function $w_+$, being majorized by a power function
of order $\gamma_+$, exaggerates the probabilities of rare events.
In particular, the probability of large portfolio returns is exaggerated.
In this way, for large portfolio values, the distortion counteracts
the risk-aversion expressed by $u_+$, which is majorized
by a concave power function $x^{\alpha_+}$.
These observations explain the appearance
of the term $\alpha_+/\gamma_+$ in \eqref{a} as ``risk aversion of the
agent on large gains modulated by her distortion function''.
Note that the agent will have a maximal risk aversion in the modified sense if (i)
$\alpha_+$ is high, i.e. close to $1$ and (ii)
$\gamma_+$ is low i.e. close to $0$ (for small value of $\gamma_+$ the agent distorts a lot the probability of rare events and, in particular,
of large gains).  Thus in \eqref{a}
we stipulate that this modulated risk-aversion parameter should still
be outbalanced by the loss aversion of the investor (as represented by parameter
$\alpha_-$ coming from the majorant of $u_-$).

A similar interpretation for the term $\alpha_-/\gamma_-$ in \eqref{bulb}
can be given. One may hope that \eqref{a} could eventually be weakened
to \eqref{bulb}. We leave the exploration of
this for future research.

We also note that the functions in Example \ref{extraa} satisfy
Assumption \ref{as1} whenever \eqref{a} holds.}

%Finally, note that our setup apply to a large class of utility
%functions. Consider any strictly increasing and concave functions
%$u_{\pm}$ such that $u_{\pm}(0)=0$ and the third derivative of $u_-$
% is strictly positive (this is in particular the case of the exponential utility
% function $u_-(x)=1-e^{-ax}$). Then as $u_+(x) \leq u_+(0) +
%u_+'(0)x$ and $u_-(x) \geq u_-(0)-u''_-(0)\frac{x^2}2$,
%\eqref{tenet} and \eqref{tenet-} are satisfied and \eqref{a} holds
%whenever $\gamma_+>1/2$.

\end{remark}

The assumption below requires
that the reference point $B$ should be comparable to the market performance in the sense that it can be sub-hedged by some
portfolio strategy $\phi\in\Phi$.

\begin{assumption}\label{as2}
We fix a scalar
random variable $B$ such that, for some strategy $\phi\in \Phi$ and for
some $b\in\mathbb{R}$, we have
\begin{equation}\label{home}
X_T^{b,\phi}=b+\sum_{t=1}^{T} \phi_t\Delta S_t\leq B.%\label{troismages}
\end{equation}
\end{assumption}

The main result of the present section is the following.
\begin{theorem}\label{ajtothm}
Under Assumptions \ref{marche}, \ref{as1} and \ref{as2},
\[
\sup_{\theta\in\mathcal{A}(X_0)}V(X_0;\theta_1,\ldots,\theta_T)<\infty,
\]
whenever $X_0\in \Xi_0^1$ with $E\vert X_0\vert^{\alpha_-}<\infty$.
\end{theorem}

In particular, the result applies for $X_0$ a deterministic constant.

Now we sketch the strategy adopted for proving the well-posedness result of Theorem \ref{ajtothm}.
First, we introduce an expected utility objective $\tilde{V}$ that dominates the behavioural objective $V$ (see Lemma \ref{eel} and Definition \ref{auxi}).
As the dynamic programming does not work for $V$, we do not introduce some one-period model associated to
$\tilde{V}$ as it is usually done in expected concave utility theory. Instead,
we make use of a multi-periodic auxiliary optimization problem $\tilde{V}_t$ (between $t$ and $T$: see Definition
\ref{auxit}). Then in Lemma \ref{crux}, we show by induction that starting from any strategy $(\theta_{t+1},\ldots, \theta_T)$,
it is always possible
to build a strategy $(\tilde{\theta}_{t+1},\ldots, \tilde{\theta}_T)$ which performs better for
the optimisation problem $\tilde{V}_t$ and which is bounded by a linear function of the initial capital $X_t$.
Finally, applying the fact that $\tilde{V}$ dominates $V$, we use the strategy
$(\tilde{\theta}_{1},\ldots, \tilde{\theta}_T)$ in order to prove that $V(X_0;\theta_1,\ldots,\theta_T)$
is always bounded by $D(1+E\vert X_0\vert^{\alpha_-})$, where
the constant $D$ does not dependent of $\theta$ (see \eqref{last}),
showing well-posedness of Theorem \ref{ajtothm}.

In the sequel, we will often use the following facts:  for all $x,y\in\mathbb{R}$, one has:
\begin{eqnarray*}
\vert x+y\vert^{\eta} & \leq &  \vert x\vert^{\eta}+\vert y\vert^{\eta}, \;  \mbox{ for } 0< \eta \leq 1,\\
\vert x+y\vert^{\eta} & \leq & 2^{\eta-1}(\vert x\vert^{\eta}+\vert y\vert^{\eta}), \; \mbox{ for } \eta \geq 1.
\end{eqnarray*}

\begin{lemma}\label{eel} Let Assumptions \ref{as1}, \ref{as2} hold. There exist constants $\tilde{k}_{\pm}>0$, such that for all $X_0\in \Xi_0^1$ and
$\theta \in \Phi$:
\begin{eqnarray}
%\label{lad}
\nonumber
V^+(X_0;\theta_{1},\ldots,\theta_{T}) & \leq &\tilde{k}_+E\left(1+
\vert X_0+\sum_{n=1}^T (\theta_n-\phi_n)\Delta S_n\vert^{\lambda \alpha_+ }\right)\\
%\nonumber & :=& \tilde{V}^+(X_0;\theta_1,\ldots,\theta_{T}), \\
%\label{ladd}
\nonumber
V^-(X_0;\theta_{1},\ldots,\theta_{T}) & \geq & \tilde{k}_- \left(E[X_0+\sum_{n=1}^T
(\theta_n-\phi_n)\Delta S_n
-b]_-^{\alpha_-} -1\right).
%\\ \nonumber &:=& \tilde{V}^-(X_0;\theta_1,\ldots,\theta_{T}).
\end{eqnarray}
\end{lemma}
\begin{proof}
See Appendix \ref{appendix eel}.
\end{proof}
We introduce the auxiliary optimization problem
with objective function $\tilde{V}$:
\begin{definition}
\label{auxi}
For all $X_0\in \Xi_0^1$ and
$\theta \in \Phi$, we define:
\begin{eqnarray*}
\tilde{V}^+(X_0;\theta_1,\ldots,\theta_{T})& := & \tilde{k}_+E\left(1+
\vert X_0+\sum_{n=1}^T (\theta_n-\phi_n)\Delta S_n\vert^{\lambda \alpha_+ }\right),\\
\tilde{V}^-(X_0;\theta_1,\ldots,\theta_{T}) &:= & \tilde{k}_- \left(E[X_0+\sum_{n=1}^T
(\theta_n-\phi_n)\Delta S_n
-b]_-^{\alpha_-} -1\right).
\end{eqnarray*}
For $X_0\in\Xi_0^1$, let $\tilde{\mathcal{A}}(X_0)=\{\theta\in\Phi\; |\; \tilde{V}^-(X_0;\theta_{1},\ldots,\theta_{T})<\infty\}$.
Whenever $\theta \in \tilde{\mathcal{A}}(X_0)$, we set
\begin{eqnarray*}
\tilde{V}(X_0;\theta_1,\ldots,\theta_{T}) & := & \tilde{V}^+(X_0;\theta_1,\ldots,\theta_{T})
-\tilde{V}^-(X_0;\theta_1,\ldots,\theta_{T}).
\end{eqnarray*}

\end{definition}
As no
probability distortions are involved in $\tilde{V}$, we can perform a kind of
dynamic programming on this auxiliary problem, formulated between time $t$ and $T$,
$0\leq t\leq T$.
\begin{definition}
\label{auxit}
For all $t=0,\ldots,T$, $X_t\in\Xi_t^1$ and $\theta_n\in \Xi_{n-1}^d$, $t+1 \leq n \leq T$, we set
\begin{eqnarray*}
\tilde{V}^+_t(X_t;\theta_{t+1},\ldots,\theta_{T})& := & \tilde{k}_+E\left(1+
\vert X_t+\sum_{n=t+1}^T (\theta_n-\phi_n)\Delta S_n\vert^{\lambda \alpha_+ }| \mathcal{F}_t\right),\\
\tilde{V}^-_t(X_t;\theta_{t+1},\ldots,\theta_{T}) &:= & \tilde{k}_- \left(E\left([X_t+\sum_{n=1}^T
(\theta_n-\phi_n)\Delta S_n
-b]_-^{\alpha_-} | \mathcal{F}_t\right)-1\right).
\end{eqnarray*}
For $X_t\in\Xi_t^1$, let $\tilde{\mathcal{A}}_t(X_t)=\{(\theta_{t+1},\ldots,\theta_T)\; |\;
\tilde{V}^-_t(X_t;\theta_{t+1},\ldots,\theta_{T})<\infty \,a.s.\}$.
For $(\theta_{t+1},\ldots,\theta_T)\in \tilde{\mathcal{A}}_t(X_t)$, we define
\begin{eqnarray*}
\tilde{V}_t(X_t;\theta_{t+1},\ldots,\theta_{T}) & := & \tilde{V}_t^+(X_t;\theta_{t+1},\ldots,\theta_{T})
-\tilde{V}_t^-(X_t;\theta_{t+1},\ldots,\theta_{T}).
\end{eqnarray*}

\end{definition}

%To this end we introduce, for all $t=0,\ldots,T$, $X_t\in\Xi_t^1$ and
%$\theta_n\in \Xi_{n-1}^d$, $t+1 \leq n \leq T$ the quantities
%\begin{eqnarray*}
%\tilde{V}^+_t(X_t;\theta_{t+1},\ldots,\theta_{T}):=\tilde{k}_+E\left(1+
%\vert X_t+\sum_{n=t+1}^T (\theta_n-\phi_n)\Delta S_n\vert^{\lambda \alpha_+ } | \mathcal{F}_t\right),\\
%\tilde{V}^-_t(X_t;\theta_{t+1},\ldots,\theta_{T}):=\tilde{k}_- E\left([X_t+\sum_{n=t+1}^T
%(\theta_n-\phi_n)\Delta S_n
%-b]_-^{\alpha_-} | \mathcal{F}_t\right)-\tilde{k}_-.
%\end{eqnarray*}
%Whenever $\tilde{V}^-_t(X_t;\theta_{t+1},\ldots,\theta_{T})<\infty$ a.s., we also define
%\[
%\tilde{V}_t(X_t;\theta_{t+1},\ldots,\theta_{T}):=\tilde{V}^+_t(X_t;\theta_{t+1},\ldots,\theta_{T})
%-\tilde{V}^-_t(X_t;\theta_{t+1},\ldots,\theta_{T}).
%\]
%We denote by $\tilde{\mathcal{A}}_t(X_t)$ the set of $(\theta_{t+1},\ldots,\theta_T)$ such that
%$\tilde{V}^-_t(X_t;\theta_{t+1},\ldots,\theta_{T})<\infty$ a.s.

\begin{lemma}
\label{tavasz} Under Assumptions \ref{as1} and \ref{as2},
$(\theta_{t+1},\ldots,\theta_T)\in\tilde{\mathcal{A}}_t(X_t)$ implies
$(\theta_{t+m+1},\ldots,\theta_T)\in\tilde{\mathcal{A}}_{t+m}(X_t+\sum_{n=t+1}^{t+m}
(\theta_{n}-\phi_n)\Delta S_{n})$, for $m\geq 0$.
The following inclusions also hold true:
$\mathcal{A}(X_0)\subset \tilde{\mathcal{A}}(X_0) \subset \tilde{\mathcal{A}}_0(X_0)$.
\end{lemma}
\begin{proof}
Remark that
$E\left(\tilde{V}^-_{t+m}(X_t+\sum_{n=t+1}^{t+m}
(\theta_{n}-\phi_n)\Delta S_{n};\theta_{t+m+1},\ldots,\theta_T) | {\mathcal F}_t \right)=
\tilde{V}^-_{t}(X_t;\theta_{t+1},\ldots,\theta_T)$. Recall also that for any bounded from below random
variable $X$, $E(X| {\mathcal F}_t)<\infty$ implies that $X<\infty$.
%and the fact that if the conditional expectation of a random variable is $<\infty$ then the random variable itself has to be $<\infty$.
This gives the first assertion. For the same reason,
$\tilde{\mathcal{A}}(X_0)\subset\tilde{\mathcal{A}}_0(X_0)$ and, by Lemma
\ref{eel}, $\mathcal{A}(X_0)\subset\tilde{\mathcal{A}}(X_0)$.
\end{proof}

%\begin{remark}\label{tavasz} {\rm Clearly,
%$(\theta_{t+1},\ldots,\theta_T)\in\tilde{\mathcal{A}}_t(X_t)$ implies
%$(\theta_{t+m+1},\ldots,\theta_T)\in\tilde{\mathcal{A}}_{t+m}(X_t+\sum_{n=t+1}^{t+m}
%(\theta_{n}-\phi_n)\Delta S_{n})$, for $m\geq 0$; this follows from
%the tower law for conditional expectations and the fact that a
%bounded from below and integrable random variable is almost surely
%finite. For the same reason,
%$\tilde{\mathcal{A}}(X_0)\subset\tilde{\mathcal{A}}_0(X_0)$ and, by Lemma
%\ref{eel}, $\mathcal{A}(X_0)\subset\tilde{\mathcal{A}}(X_0)$.}
%\end{remark}

The crux of our arguments is contained in the next result. It states that
each strategy in $\tilde{\mathcal{A}}_t(X_t)$ can be replaced by another
one such that the latter performs better (see \eqref{buu}) and it is
close to $\phi$ in the sense that their distance is linear in the
initial endowment $X_t$ (see \eqref{boo}).

\begin{lemma}\label{crux} Assume that Assumptions \ref{marche}, \ref{as1}, \ref{as2} hold true. Then for each $0\leq t\leq T$,
there exist $C^t_n\in\mathcal{W}_n^+$, $n=t,\ldots,T-1$, such that,
for all $X_t\in \Xi_t^1$ and  $(\theta_{t+1},\ldots,\theta_T)\in\tilde{\mathcal{A}}_t(X_t)$,
there exists $(\tilde{\theta}_{t+1},\ldots,\tilde{\theta}_T)\in\tilde{\mathcal{A}}_t(X_t)$ satisfying
for $n=t+1,\ldots,T$:
\begin{equation}\label{boo}
%\label{kezdet}
\vert\tilde{\theta}_n-\phi_n\vert\leq C^t_{n-1} [\vert X_t\vert+1],
\end{equation}
and
\begin{eqnarray}\label{buu}
\tilde{V}_t(X_t;\theta_{t+1},\ldots,\theta_T)\leq \tilde{V}_t(X_t;\tilde{\theta}_{t+1},\ldots,\tilde{\theta}_T).
\end{eqnarray}
%with $E\vert X_t\vert^{\alpha_-}<\infty$.
\end{lemma}
\begin{proof}
See Appendix \ref{appendix crux}.
\end{proof}

\begin{proof}[Proof of Theorem \ref{ajtothm}] If $\mathcal{A}(X_0)$ is empty, there is nothing to prove.
Otherwise, by Lemmas \ref{crux} and \ref{tavasz},
 there is  $C_n^0 \in \mathcal{W}^+_n$, $0\leq n \leq T-1$ such that,
for all $\theta \in \mathcal{A}(X_0) \subset \tilde{\mathcal{A}}_0(X_0)$, there exists $\tilde{\theta} \in \tilde{\mathcal{A}}_0(X_0)$
satisfying
$\vert\tilde{\theta}_n-\phi_n\vert\leq C^0_{n-1} [\vert X_0\vert+1]$, $1\leq n\leq T$ and
\[
\tilde{V}_0(X_0;\theta_{1},\ldots,\theta_T)\leq\tilde{V}_0(X_{0};\tilde{\theta}_{1},\ldots,
\tilde{\theta}_T).
\]
%using
%\eqref{kezdet},
As $\theta \in \mathcal{A}(X_0)$, by
Lemma \ref{eel}, using H\"older's inequality with $p=\alpha_-/(\lambda \alpha_+ )$ and
its conjugate number $q$ and the rough estimation $x^{1/p} \leq x+1$,
\begin{eqnarray}
\nonumber
V(X_0;\theta_{1},\ldots,\theta_T) & \leq & \tilde{V}(X_0;\theta_{1},\ldots,\theta_T)   =
E\tilde{V}_0(X_0;{\theta}_{1},\ldots,{\theta}_T)
 \leq   E\tilde{V}_0(X_0;\tilde{\theta}_1,\ldots,\tilde{\theta}_T) \\
\nonumber  & \leq &   E\tilde{V}^+_0(X_0;\tilde{\theta}_{1},\ldots,\tilde{\theta}_T) \\
\nonumber & \leq & \tilde{k}_+E\left(1+ |X_0|^{\lambda \alpha_+ }+ \sum_{n=1}^T |\tilde{\theta}_n -{\phi}_n|^{\lambda \alpha_+ }
 |\Delta S_n|^{\lambda \alpha_+ } \right)\\
 \nonumber & \leq & \tilde{k}_+E\left(\left(1+ |X_0|^{\lambda \alpha_+ }\right)
 \left(1+ \sum_{n=1}^T (C_{n-1}^0)^{\lambda \alpha_+ }
 |\Delta S_n|^{\lambda \alpha_+ } \right)\right)\\
 \nonumber
& \leq &  \tilde{k}_+2^{\frac{p-1}{p}} E^{1/p}(1+\vert X_0\vert^{\alpha_-}) E^{1/q}\left(1+ \sum_{n=1}^T (C_{n-1}^0)^{\lambda \alpha_+ }
 |\Delta S_n|^{\lambda \alpha_+ } \right)^q\\
  \nonumber
 & \leq &  \tilde{k}_+2^{\frac{p-1}{p}} (2+E\vert X_0\vert^{\alpha_-}) E^{1/q}\left(1+ \sum_{n=1}^T (C_{n-1}^0)^{\lambda \alpha_+ }
 |\Delta S_n|^{\lambda \alpha_+ } \right)^q\\
 \label{last}
   & \leq &  D (1+E\vert X_0\vert^{\alpha_-}),
\end{eqnarray}
for an appropriate constant $D$ (independent of $\theta$), noting that $\mathcal{W}$ is closed under addition and multiplication. As $E|X_0|^{\alpha_-}<\infty$ was
assumed, we get that
this expression is finite, showing Theorem \ref{ajtothm}.
%$C_t(\vert X_t\vert^{\lambda \alpha_+ } +1)$ also belongs to $\mathcal{W}^+_t$
%and is thus finite a.s, proving the Theorem.
\end{proof}

\begin{remark}
{\rm Theorem 3.2 of \cite{jz} states, in a continuous-time context, that in a typical Brownian market model
our optimization problem is ill-posed whenever $u_+$ is unbounded and
$w_-(p)=p $ (i.e. no distortion on losses).}

{\rm It is worth contrasting this with Theorem \ref{ajtothm} above which states
that even if $w_-(p)=p$ and $u_+(x)\sim x^{\alpha_+}$ the problem is
well-posed provided that $\alpha_+/\gamma_+<\alpha_-$.}

{\rm This shows that discrete-time models behave slightly differently from
their continuous-time counterparts as far as well-posedness is concerned.
%.phenomenon stems from the particularity of continuous-time models where
%richness of th model: essentially
%``every'' random variable can be replicated and
In discrete-time models the terminal values of admissible portfolios
%are generically incomplete (the complete cases can be reduced to a finite $\Omega$,
%see \cite{js}, where well-posedness holds by the ``smallness'' of $\Omega$).
%As
form a relatively small family of random variables hence ill-posedness does not occur even in cases where it does in the continuous-time setting, where
the set of attainable payoffs is much richer.}
%even if $w_-$ is the identity (as long as the other assumptions of
%Theorem \ref{ajtothm} hold).}
\end{remark}

For the subsequent sections we need to extend and refine the arguments
of Lemma \ref{crux} (see \eqref{cars} versus \eqref{figula} below). This is done in the following lemma.

{
\begin{lemma}\label{oslo} Let Assumptions \ref{marche}, \ref{as1}, \ref{as2}
be in force.
Fix $c\in\mathbb{R}$ and $\iota ,o$ satisfying $\lambda\alpha_+<\iota
<o<\alpha_-$.
Then there exists $K_t$
such that
$$E\vert \theta_{t+1}-\phi_{t+1}\vert^{\iota }\leq K_t[E\vert X_t\vert^{o}+1],$$
for any $X_t \in \Xi_t^1$ with $E\vert
X_t\vert^{o}<\infty$ and  $(\theta_{t+1}, \ldots,
\theta_{T})\in\tilde{\mathcal{A}}_t(X_t)$.
as soon as
\[
E\tilde{V}_t(X_t;\theta_{t+1},\ldots,\theta_T)\geq c.
\]
Note that the constant $K_t$ do not depend either on $X_t$ or $\theta$.
\end{lemma}
}
%\begin{lemma}\label{oslo} Let Assumptions \ref{marche}, \ref{as1}, \ref{as2}
%be in force.
%Fix $c\in\mathbb{R}$ and $\iota ,o$ satisfying $\lambda\alpha_+<\iota
%<o<\alpha_-$. Assume that
%\[
%E\tilde{V}_t(X_t;\theta_{t+1},\ldots,\theta_T)\geq c
%\]
%for  $X_t \in \Xi_t^1$ with $E\vert
%X_t\vert^{o}<\infty$ and for some $(\theta_{t+1}, \ldots,
%\theta_{T})\in\tilde{\mathcal{A}}_t(X_t)$. Then there exists $K_t$
%such that
%$$E\vert \theta_{t+1}-\phi_{t+1}\vert^{\iota }\leq K_t[E\vert X_t\vert^{o}+1],$$
%where $K_t$ does not depend on either $X_t$ or $\theta$.
%\end{lemma}
\begin{proof}
See Appendix \ref{appendix oslo}.
\end{proof}

\section{On the class of admissible strategies}\label{beszur}

In this section we look at some unexpected phenomena that arise when investigating
the existence of an optimal strategy for problem \eqref{plumduff}}.
%\eqref{home}.

In the context of game theory it was suggested already in \cite{borel}
to apply mixed strategies (i.e.
ones using randomness) as opposed to
pure strategies (i.e. ones without randomness). This relaxation
of the set of strategies is indispensable for cornerstone results such as the minimax and equilibrium
theorems to hold (see \cite{von} and \cite{nash}).
These celebrated theorems led to a
widespread application of game theory in economics.

It is important to note that at the beginning, when the basic notions  of game
theory were introduced, there was no associated randomness appearing in the
problem formulation. The randomness hence did not come from the nature of the considered problem but it was introduced  exogenously so that a satisfactory theory could be established.

In the context of optimal stochastic control for partially observed diffusions, auxiliary randomness has been used in order to prove the existence of an optimal control. \cite{fleming_pardoux} and \cite{bkr} have showed, in different setting, that the optimal control
fails to exist unless a relaxed class of randomized controls (called wide-sense) is used\footnote{We thank Ioannis Karatzas for drawing our attention to the references \cite{fleming_pardoux} and \cite{bkr}.}.
%In the setting of stochastic control the use of external randomness
%(that is not related to the controlled stochastic system) has been studied and was found to give better performance in certain cases, see e.g.
%\cite{fleming_pardoux}. It may even happen that an optimal strategy
%fails to exist unless a relaxed class of randomized strategies is used,
%see \cite{bkr}\footnote{We thank Ioannis Karatzas for drawing our attention
%to the references \cite{fleming_pardoux} and \cite{bkr}.}.
As far as we know, in the optimal
investment context there has been no such investigations yet. For this reason we explain what we mean by external randomness
in the framework of the present article.

A portfolio strategy $\theta_t$ at time $t$ is random by nature, it is a function of the information up to $t-1$, as encoded by
$\mathcal{F}_{t-1}$. One may ask, inspired by game theory and
\cite{fleming_pardoux}, whether it makes
sense to add further randomization to the strategy that is not intrinsic to
the problem but comes from an exogenous random source.
In more concrete terms, is it worth taking $\varepsilon$
independent of the whole history $\mathcal{F}_T$, and considering
$\theta_t$ that is a function of $\mathcal{F}_{t-1}$ and $\varepsilon$ ?
The practical implementation of such an idea would be easy: a computer
may be used to generate the random number $\varepsilon$.

As far as we know, this idea never came up in utility theory because
in the standard framework it has not been used for existence results and also it
does not lead to a higher level of satisfaction for the agent. To see
this, consider a utility function $u:\mathbb{R}\to\mathbb{R}$.
Assume for simplicity that $T=1$, $\mathcal{F}_0=\{\emptyset,\Omega\}$
and $\mathcal{F}_1:=\sigma(\Delta S_1)$. Fix $X_0\in\mathbb{R}$.
We assume here that the family of admissible strategies is the set of $\mathcal{F}_0$-measurable random variables, i.e.
$\mathcal{A}:=\mathbb{R}$.
We assume also that $Eu(X_0+\theta\Delta S_1)$
is finite for all $\theta\in \mathcal{A}$ and that an optimal
investment $\theta^*\in \mathcal{A}$ exists, i.e. $Eu(X_0+\theta^*\Delta S_1)=
\sup_{\theta\in\mathcal{A}}Eu(X^{X_0,\theta}_T)$ (see \cite{RS05} for conditions on $u$ and $S$ ensuring that the problem is well-posed and admits some solution).

%We assume that $Eu(X_0+\theta\Delta S_1)$
%is well-defined and finite for all $\theta\in\mathcal{A}$ and that there
%is $\theta^*\in\mathcal{A}$ such that $Eu(X_0+\theta^*\Delta S_1)=
%\sup_{\theta\in\mathcal{A}}Eu(X^{X_0,\theta}_T)$ (i.e. an optimal
%investment exists).

%Maximize
%$Eu(X_0+\theta\Delta S_1)=Eu(X^{X_0,\theta}_T)$ where $\theta$ runs over the set

Let us now define $\mathcal{F}_0':=\sigma(\varepsilon)$ with $\varepsilon$ independent of $\mathcal{F}_1$
and consider
$$\mathcal{A}':=\{\theta: \theta\mbox{ is }\mathcal{F}_0'-\mbox{measurable and }
E[u(X_0+\theta\Delta S_1)]_-<\infty\}.$$
Here it is necessary to constrain the family of $\theta$s by an integrability
condition as it may easily happen that both $E[u(X_0+\theta\Delta S_1)]_-$
and $E[u(X_0+\theta\Delta S_1)]_+$ are infinite and the expected utility
may not be defined.

We claim that
\[
\sup_{\phi\in\mathcal{A}}Eu(X_0+\phi\Delta S_1)=\sup_{\phi\in\mathcal{A}'}
Eu(X_0+\phi\Delta S_1).
\]
Indeed, $\leq$ is trivial from $\mathcal{A}\subset\mathcal{A}'$. Taking $\theta\in\mathcal{A}'$, we see
that, using the tower law and the independence of $\Delta S_1$ and $\theta$,
\begin{eqnarray*}\label{loch}
Eu(X_0+\theta\Delta S_1) & =& E\left[E[u(X_0+\theta\Delta S_1)\vert\theta]\right]\\
 & = & \int_{\mathbb{R}} E[u(X_0+t\Delta S_1)|\theta=t] P_{\theta}(dt)=
\int_{\mathbb{R}} E[u(X_0+t\Delta S_1)]P_{\theta}(dt)\\
&\leq& \int_{\mathbb{R}} E[u(X_0+\theta^*\Delta S_1)]P_{\theta}(dt)=
Eu(X_0+\theta^*\Delta S_1)\\
&=&\sup_{\phi\in\mathcal{A}} Eu(X_0+\phi\Delta S_1),
\end{eqnarray*}
where $P_{\theta}$ denotes the law of $\theta$.

This computation supports our claim that the exogenous random source $\varepsilon$ does not improve the agent's satisfaction for an expected utility criterion. Consequently, such randomizations do not make
sense and thus were never considered. Note, however, that the previous
argument relies on the tower law which applies only because we face a
criterion of expected utility.

In the setting of the present paper there are ``nonlinear'' expectations (Choquet
integrals) and it is not obvious whether an exogeneous random source is useful. In the rest of this section we will see that, somewhat surprisingly, such
randomization does improve the satisfaction of a behavioural investor and hence it is
worth exploiting.

In the rest of this section we investigate a one-step example
($T=1$) with $S_0=0$, $P(\Delta S_1=1)=P(\Delta S_1=-1)=1/2$. Set
$\mathcal{G}_0:=\{\emptyset,\Omega\}$ and
$\mathcal{G}_1:=\sigma(\Delta S_1)$. Let $\epsilon_i$, $i\geq 1$ be
a sequence of i.i.d. random variables, independent of
$\mathcal{G}_1$ such that $P(\epsilon_1=1)=P(\epsilon_1=-1)=1/2$.
Define the sigma-algebras $\mathcal{H}_0:=\{\emptyset,\Omega\}$ and
$\mathcal{H}_n:=\sigma(\epsilon_1,\ldots,\epsilon_n)$ {for $n \geq 1$}. Let
$\mathcal{A}_n$ denote the set of $\mathcal{H}_n$-measurable scalar
random variables {for $n \geq 0$}.

Fix $n{\geq 0}$ and add some external randomization
in the filtration, i.e. $\mathcal{F}_t=\mathcal{G}_t\vee\mathcal{H}_n$.
So $\Phi=\mathcal{A}_n$ in this case.
We take the initial capital $X_0=0$ and
also $B=0$. Assume that $u_+(x)={x}^{1/4}$, $u_-(x)=x$; $w_+(p)=\sqrt{p}$,
$w_-(p)=p$.
Using again the tower law and the independence of $\Delta S_1$ and $\theta$, it is easy to see that
\begin{eqnarray}
\label{rubis}
V^+(0;\theta)& = & \sqrt{\frac12}\int_0^{\infty}\sqrt{P(|\theta|^{1/4}\geq y)}dy,\\
\label{saphir}
V^-(0;\theta)& = & \frac12E|\theta|.
\end{eqnarray}
Here $\mathcal{A}(0)=\mathcal{A}_n$ because $\mathcal{H}_n$ is generated
by finitely many atoms, so there is no need for integrability restrictions
on strategies (see \eqref{saphir}).
Consider a sequence of optimization problems:
\[
(\mathcal{M}_n)\qquad\sup_{\theta\in\mathcal{A}_n}V(0;\theta).
\]
%where
%\begin{eqnarray*}
%V_+(\theta) &=& \int_0^{\infty}P(u_+([\theta\Delta S_1]_+)\geq x)dx,\\
%V_-(\theta) &=& \int_0^{\infty}P(u_-([\theta\Delta S_1]_-)\geq x)dx,\\
%V(\theta)&=&V_+(\theta)-V_-(\theta),
%\end{eqnarray*}
%and
%$\mathcal{A}_n$ denotes the set of $\mathcal{H}_n$-measurable
%scalar  random variables (as $\mathcal{H}_n$ is generated
%by finitely many atoms there is no need for integrability restrictions
%on strategies).
Introduce the notation
$$
M_n:=\sup_{\theta\in\mathcal{A}_n}V(0;\theta),\ n\geq 0.
$$

\begin{lemma}\label{ret}
The strategy $\theta\equiv 0$ is not optimal for $\mathcal{M}_n$, for any
$n$.
\end{lemma}
\begin{proof} From \eqref{rubis} and \eqref{saphir}, we get that for $\theta\in\mathcal{A}_0$, $\theta>0$,
$
V(0;\theta)=\sqrt{\frac{1}{2}}{\theta}^{1/4}-\frac{1}{2}\theta.
$
{So for $\theta>0$ small enough $V(0;\theta)$ is strictly greater than $0=V(0;0)$, showing that
$0$ is not optimal for $\mathcal{M}_0$. Hence, since $M_{n+1} \geq M_n$ for all $ n\geq 0$, it is
not optimal for neither of the $\mathcal{M}_n$}.
\end{proof}

\begin{proposition}\label{kh}
We have $M_n<M_{n+1}<\infty$ for all $n\geq 0$.\footnote{We thank Andrea Meireles
for numerically checking $M_0<M_1$, which eventually lead to the formulation of this proposition.}
\end{proposition}
\begin{proof}
See Appendix \ref{appendixkh}.
\end{proof}

Proposition \ref{kh} shows that introducing more and more external
random sources {increases} the attainable satisfaction level, i.e.
``gambling'' leads to higher agent satisfaction.  Once we accept the
hypothesis that agents act according to preferences involving the
distortions $w_{\pm}$, we also have to accept that, using external
randomness, they may (and do) increase their satisfaction level. It
seems thus reasonable to use the whole sequence $\epsilon_i$, i.e.
to optimize over $\sigma(\epsilon_i,i\in\mathbb{N})$-measurable
$\theta$ (note that in this case one needs to restrict the domain of
maximization to those $\theta$ for which $V^-(0;\theta)$ is finite,
but this is a minor point which is not crucial for our discussion).

As by Kuratowski's theorem the spaces $\{1,-1\}^{\mathbb{N}}$ and
$[0,1]$ are Borel-isomorphic (see Theorem 80 on p.159 of
\cite{dm}), one may take, instead of
$\sigma(\epsilon_i,i\in\mathbb{N})$, $\mathcal{H}=
\sigma(\varepsilon)$ where $\varepsilon$ is uniform on $[0,1]$ and
independent of $\mathcal{F}_1$. One may try to push this further by
considering a sigma-algebra generated by a sequence of independent
uniform random variables but this does not lead to a larger class of
trading strategies as $[0,1]^{\mathbb{N}}$ is Borel-isomorphic to
$[0,1]$, see again p. 159 of \cite{dm}. Finally, one may think of
extending the optimization problem to $\sigma(\epsilon_i,i\in
I)$-measurable $\theta$, for an uncountable collection $I$. This,
again, does not lead to a larger domain of optimization since if
$\theta$ is a given $\sigma(\epsilon_i,i\in I)$-measurable random
variable then it is also $\sigma(\epsilon_i,i\in I_0)$-measurable
for some countable $I_0\subset I$.\footnote{Indeed, it is enough to
show that every bounded $\sigma(\epsilon_i,i\in I)$-measurable
$\theta$ is, in fact, $\sigma(\varepsilon_i,i\in I_0)$-measurable
for some countable $I_0\subset I$ (which may depend on $\theta$).
For any $\sigma$-field ${\cal B}$, let $b({\cal B})$ be the set of
${\cal B}$-measurable, bounded and real-valued functions. Let
$\mathcal{H}:=\bigcup_{J \mbox{ countable, } J \subset
I}b(\sigma(\epsilon_j,j\in J)).$ It is easy to see  that
$\sigma(\mathcal{H})=\sigma(\epsilon_i,i\in I)$. As $\mathcal{H}$ is
clearly a monotone vector space as well as a multiplicative class,
by the Monotone Class Theorem, $b(\sigma(\mathcal{H}))=\mathcal{H}$
(see p.7 of \cite{protter}). So $b(\sigma(\epsilon_i,i\in
I))=\mathcal{H}$ and the result is proved. }

{The arguments of the previous paragraph show, together with Proposition \ref{kh},
that a natural maximal domain of optimization is ${\mathcal{H}}$.}
By using a uniform $\epsilon$ (independent of $\mathcal{F}_1$) for randomizing
the strategies an investor can increase her satisfaction and further randomizations
are pointless as they do not provide additional satisfaction.

Based on discussions of this section we reformulate the problem of existence
by enlarging the filtration: let $\mathcal{G}_t$, $t=0,\ldots,T$ be a filtration and
let $S_t$, $t=0,\ldots, T$ be a
$\mathcal{G}_t$-adapted process. Furthermore, $\mathcal{F}_t=\mathcal{G}_t\vee\mathcal{F}_0$, $t\geq 0$, where $\mathcal{F}_0=\sigma({\varepsilon})$ with ${\varepsilon}$ uniformly distributed on $[0,1]$ and independent of
$\mathcal{G}_T$.

We are now seeking $\theta^*\in
{\mathcal{A}}(X_0)$ such that
\[
V(X_0;\theta^*_1,\ldots,\theta^*_T)=\sup_{\theta\in {\mathcal{A}}(X_0)}
V(X_0;\theta_1,\ldots,\theta_T),
\]
where $V(\cdot)$ and $\mathcal{A}(X_0)$ are as defined in section
\ref{ketto}.
We will see in the next section that in this relaxed class of
randomized strategies there exists indeed an optimal
strategy.
%under reasonable conditions.

\section{Existence of an optimizer using relaxed strategies}\label{ot}

%In this section we rely on an even stronger concept of absence of arbitrage.
%\begin{assumption}
%Condition \eqref{rrr} holds true with $\kappa_t,\pi_t>0$ constants, $t=0,\ldots, T-1$.
%\end{assumption}

%In the sequel, we need to assume that the probability space we are working on is
%``rich enough''.

%Throughout this section, Assumptions \ref{marche}, \ref{as1}, \ref{as2}, \ref{as3} and \ref{as4} will be in force.

In this section we prove the existence of optimal strategies after introducing
some hypotheses. First, we will need a certain structural assumption
on the filtration.

\begin{assumption}\label{as3}
Let $\mathcal{G}_0=\{\emptyset,\Omega\}$, $\mathcal{G}_t=\sigma(Z_1,\ldots,Z_t)$ for $1\leq t\leq T$,
where the $Z_i$, $i=1,\ldots, T$ are $\mathbb{R}^N$-valued
independent random variables. $S_0$ is constant
and $\Delta S_t$ is a
continuous function of $(Z_1,\ldots,Z_t)$, for all $t\geq 1$ (hence $S_t$
is $\mathcal{G}_t$-adapted).

Furthermore, $\mathcal{F}_t=\mathcal{G}_t\vee\mathcal{F}_0$, $t\geq 0$,
where $\mathcal{F}_0=\sigma({\varepsilon})$ with ${\varepsilon}$
uniformly distributed on $[0,1]$ and independent of $(Z_1,\ldots,Z_T)$.
\end{assumption}

We may think that $\mathcal{G}_t$ contains the information available at time $t$ (given by the
observable stochastic factors $Z_i$, $i=1,\ldots,t$)
and $\mathcal{F}_0$ provides the independent random source that we use
to randomize our trading strategies as discussed in the previous section
in much detail. The random variables $Z_i$ represent the
``innovation'': the information surplus of $\mathcal{F}_i$
with respect to $\mathcal{F}_{i-1}$, in an independent way.

For the construction of optimal strategies we use weak convergence techniques which
exploit the additional randomness provided
by ${\varepsilon}$ (the situation is somewhat analogous to the construction of
a weak solution for a stochastic differential equation). Assumption \ref{as3}
holds in many cases, see section \ref{exaoa} for examples.
It may nevertheless seem that Assumption \ref{as3} is quite restrictive. In particular,
it would be desirable to weaken the independence assumption on the $Z_i$.
For this reason
we propose another assumption which may be easier to check in certain
model classes and which will be shown to imply Assumption \ref{as3}.

\begin{assumption}\label{masikk}
 Let $\mathcal{G}_0=\{\emptyset,\Omega\}$, $\mathcal{G}_t=\sigma(\tilde{Z}_1,\ldots,\tilde{Z}_t)$ for
$1\leq t\leq T$,
where the $\tilde{Z_i}$, $i=1,\ldots,T$ are $\mathbb{R}^N$-valued
random variables with a continuous and everywhere positive
joint density $f$ on $\mathbb{R}^{TN}$ such that for all $i=1,\ldots,TN$,
the function
\begin{equation}\label{torto}
z\to \sup_{x^1,\ldots,x^{i-1}} f_i(x^1,\ldots,x^{i-1},z)
\end{equation}
is integrable on $\mathbb{R}$,  where $f_i$ is the marginal
density of $f$ with respect to its first $i$ coordinates, {for $i=2,\ldots,TN$.}
$S_0$ is constant
and $\Delta S_t$ is a
continuous function of $(\tilde{Z}_1,\ldots,\tilde{Z}_t)$, for all $t\geq 1$.

Furthermore, $\mathcal{F}_t=\mathcal{G}_t\vee\mathcal{F}_0$, $t\geq 0$,
where $\mathcal{F}_0=\sigma({\varepsilon})$ with ${\varepsilon}$
uniformly distributed on $[0,1]$ and independent of $(\tilde{Z}_1,\ldots,\tilde{Z}_T)$.
\end{assumption}

\begin{remark}
{\rm Condition \eqref{torto} is quite weak, it holds, for example, when there is $C>0$
such that
\[
f(\mathbf{x})\leq {C}\prod_{i=1}^{TN} g_{i}(\mathbf{x}^i),\quad \mathbf{x}\in\mathbb{R}^{TN},
\]
for some {positive}, bounded and integrable (on $\mathbb{R}$) functions $g_{i}$,
({for example $g_i(y)=1/(1+y^2)$}).
}
\end{remark}

\begin{proposition}\label{+++}
If Assumption \ref{masikk} above holds true then so does Assumption \ref{as3}.
\end{proposition}
\begin{proof}
See Appendix \ref{+++appendix}.
\end{proof}

Furthermore, the following assumption on continuity and on the
initial endowment is imposed.

\begin{assumption}\label{as4}
The random variable $B$ is a continuous function of $(Z_1,\ldots,Z_{T})$, $X_0$ is
deterministic and $\mathcal{A}(X_0)$
is not empty. $u_{\pm},w_{\pm}$ are continuous functions.
%\noindent $X_0$ is such that $E(X_0^{\alpha_-})<\infty$ and
%be deterministic and $Eu_-([X_0-B]_-)<\infty$.
\end{assumption}

\begin{remark} {\rm If $B$ is a continuous function of $(S_0,\ldots,S_T)$ then Assumption
\ref{as3} clearly implies the first part of Assumption \ref{as4}. For conditions
implying $\mathcal{A}(X_0)\neq\emptyset$ see Remark \ref{qtya} above. }
\end{remark}

\begin{remark}\label{lublin}
{\rm We may and will suppose that the $Z_i$ figuring in Assumption
\ref{as3} are bounded. This can always
be achieved by replacing each coordinate $Z_i^j$ of $Z_i$ with
$\mathrm{arctan}\, Z_i^j$ for $j=1,\ldots, N$, $i=1,\ldots,T$.}
\end{remark}

We now present our main result on the existence of an optimal strategy.

\begin{theorem}\label{maine} Let Assumptions \ref{marche}, \ref{as1}, \ref{as2}, \ref{as3} and \ref{as4}
hold. Then there is $\theta^*\in\mathcal{A}(X_0)$ such that
\[
V(X_0;\theta^*_1,\ldots,\theta^*_T)=\sup_{\theta\in\mathcal{A}(X_0)}V(X_0;\theta_1,\ldots,\theta_T)<\infty.
\]
\end{theorem}
We sketch the proof of Theorem \ref{maine}.
First, we fix some
$\lambda \alpha_+ <\chi<\alpha_-$ for what follows.
In Lemma \ref{oslo} above and Lemma \ref{porto} below, we refine certain
arguments of Lemma \ref{crux}: instead of building
a particular strategy $\tilde{\theta}$ from some strategy $\theta$, we show that the boundness
from below of $\tilde{V}$ implies that $\sup_{\theta,t} E\vert \theta_{t+1}-\phi_{t+1}\vert^{\tau }$
is bounded, for any $\chi <\tau < \alpha_-$. This allows us to prove that a maximizing sequence for problem $V$ is tight and thus
weakly converges.
The problem is then to construct some strategy with the same law as the above weak limit but also
$\mathcal{F}$-predictable. To this end we need first to consider a sequence including
$(Z_1,\ldots,Z_T)$. But this is not enough: consider the following example {showing that a weak
limit of some $\mathcal{F}$-predictable sequence may fail to be $\mathcal{F}$-predictable}.

\begin{example}
{\rm Consider the probability space $\Omega:=[0,1]$ equipped with its Borel
sigma-field and the Lebesgue measure. Take $\xi(\omega):=\omega$ for
$\omega\in\Omega$ and for $n\geq 1$,
$\eta_n(\omega):=n(\omega-(k/n))$, for
$\omega\in [k/n,(k+1)/n)$, $k=0,\ldots,n-1$. Clearly, each $\eta_n$
is a function of $\xi$, actually, every random variable on this
probability space is a function of $\xi$. Nonetheless the weak limit of the sequence
$\mathrm{Law}(\xi,\eta_n)$ is easily seen to be the uniform law on $[0,1]^2$.
We claim  that there is no $\eta$ defined on $\Omega$ such that $(\xi,\eta)$
has uniform law on $[0,1]^2$. Indeed, $\eta$ is necessarily a function
of $\xi$ hence it cannot be also independent of it. This shows that, in
order to construct $\eta$ with the property that $(\xi,\eta)$ has the
required (uniform) law, one needs to extend the probability space.}
\end{example}

%Let $X$ be a uniformly
%distributed on $[-1,1]$ random variable and $X_n=\sin(nX)$. Then $(X,X_n)$ converges weakly to
%a uniformly distributed on $[-1,1]^2$ random variable.
Therefore we add some random noises $(\varepsilon',
\varepsilon_1,\ldots,\varepsilon_T)$. The noise $\varepsilon'$ is
used to build some admissible set $\mathcal{A}'(X_0)$, where we
choose some maximizing sequence
$(\theta_1(j),\ldots,\theta_T(j))_j$. Then we consider the sequence
$(Y_j)_j$, where
$Y_j:=(\varepsilon',\theta_1(j),\ldots,\theta_T(j),Z_1,\ldots,Z_T)$,
which is also tight and call $\mu$ its weak limit. Then we construct
inductively, $\theta^*_t$, $t=1,\ldots,T$ such that
$(\varepsilon',\theta^*_1,\ldots,\theta^*_T,$ $Z_1,\ldots,Z_T)$ has
law $\mu$  and $\theta^*_t$ depends only on $(\varepsilon',
\varepsilon_1,\ldots,\varepsilon_{t}, Z_1,\ldots,Z_{t-1})$ and hence
it is $\mathcal{F}_{t-1}$-measurable. Finally, we show that this
strategy $\theta^*$ is optimal.

{
\begin{lemma}\label{porto} Let Assumptions \ref{marche}, \ref{as1}, \ref{as2}
be in force.
Fix $c\in\mathbb{R}$ and $\tau$ with
${\lambda \alpha_+} <\tau<\alpha_-$.
Then there exist constants $G_t,t=0,\ldots,T-1$
such that
$$E\vert \theta_{t+1}-\phi_{t+1}\vert^{\tau}\leq
G_t[E\vert X_0\vert^{\alpha_-}+1] \mbox{ for $t=0,\ldots,T-1$ } $$
for any $\theta\in\tilde{\mathcal{A}}(X_0)$ and $X_0\in\Xi^1_0$ with $E\vert X_0\vert^{\alpha_-}<\infty$
such that
$$
\tilde{V}(X_0;\theta_{1},\ldots,\theta_T)=E\tilde{V}_0(X_0;\theta_{1},\ldots,\theta_T)\geq c.$$
Note that the constants $G_t,t=0,\ldots,T-1$ do not depend either on $X_0$ or on $\theta$.
\end{lemma}
}
%\begin{lemma}\label{porto} Let Assumptions \ref{marche}, \ref{as1}, \ref{as2}
%be in force.
%Fix $c\in\mathbb{R}$ and assume that
%\[
%\tilde{V}(X_0;\theta_{1},\ldots,\theta_T)=E\tilde{V}_0(X_0;\theta_{1},\ldots,\theta_T)\geq c
%\]
%for some $X_0\in\Xi^1_0$ with $E\vert X_0\vert^{\alpha_-}<\infty$
%and some $\theta\in\tilde{\mathcal{A}}(X_0)$. Fix $\tau$ with
%$\chi<\tau<\alpha_-$. There exist constants $G_t,t=0,\ldots,T-1$
%such that $E\vert \theta_{t+1}-\phi_{t+1}\vert^{\tau}\leq
%G_t[E\vert X_0\vert^{\alpha_-}+1]$ for $t=0,\ldots,T-1$, and $G_t$
%do not depend on $X_0$ or on $\theta$.
%\end{lemma}

\begin{proof}
See Appendix \ref{appendix porto}.
\end{proof}

\begin{proof}[Proof of Theorem \ref{maine}]
Lemma \ref{ode} with the choice $E:={\varepsilon}$, $l=2$ gives us
$\tilde{\varepsilon},\varepsilon'$ independent, uniformly
distributed on $[0,1]$ and $\mathcal{F}_0$-measurable. Introduce
\[
\mathcal{A}'(X_0):=\{\theta\in\mathcal{A}(X_0):\theta_t\mbox{ is }\mathcal{F}_{t-1}'\mbox{-measurable for
all } t=1,\ldots,T\},
\]
where $\mathcal{F}_t':=\mathcal{G}_t\vee \sigma(\varepsilon')$. Note
that if $\theta\in\mathcal{A}(X_0)$ then there exists
$\theta'\in\mathcal{A}'(X_0)$ such that the law of $(\theta,\Delta
S)$ equals that of $(\theta',\Delta S)$ (since the law of
${\varepsilon}$ equals that of $\varepsilon'$ and both are
independent of $\Delta S$). It follows that for all
$\theta\in\mathcal{A}(X_0)$ there is $\theta'\in\mathcal{A}'(X_0)$
with
\[
V(X_0;\theta_1,\ldots,\theta_T)=V(X_0;\theta_1',\ldots,\theta_T').
\]
%e.g. \cite{d-m}).
%The condition $Eu_-([X_0-B]_-)<\infty$ in Assumption \ref{as4} guarantees
%that the identically $0$ strategy is in $\mathcal{A}_0(X_0)$, in particular,
%$\mathcal{A}_0(X_0)\neq\emptyset$.

%Surely, $\phi\in\mathcal{A}_0(X_0)$ with $V(X_0;\phi_1,\ldots,\phi_T)\geq
%Eu_+((X_0-b)_+)-Eu_-((X_0-b)_-)>-\infty$,hence the supremum is taken over a nonempty set.

\noindent Take $\theta(j)\in\mathcal{A}(X_0),j\in\mathbb{N}$ such that
\[
V(X_0;\theta_1(j),\ldots,\theta_T(j))\to \sup_{\theta\in\mathcal{A}(X_0)}V(X_0;\theta_1,\ldots,\theta_T),
\quad j\to\infty.
\]
By Assumption \ref{as4} and Theorem \ref{ajtothm}, the supremum is finite and we
can fix $c$ such that
$
-\infty<c<\inf_j V(X_0;\theta_1(j),\ldots,\theta_T(j))$. By Lemma \ref{eel} it implies that for
all $j$,
$$\tilde{V}(X_0;\theta_1(j),\ldots,\theta_T(j))>c.$$
By the discussions above we may and will assume
$\theta(j)\in\mathcal{A}'(X_0),j\in\mathbb{N}$. Apply Lemma
\ref{porto} for some $\tau$ such that $\chi<\tau < \alpha_-$
to get
$$
\sup_{j,t} E\vert\theta_t(j)-\phi_t\vert^{\tau}<\infty.
$$
It follows that the
sequence of $T(d+N)+1$-dimensional
random variables
\[
\tilde{Y}_j:=(\varepsilon',\theta_1(j)-\phi_1,\ldots,\theta_T(j)-\phi_T,Z_1,\ldots,Z_T)
\]
are bounded in $L_{\tau}$ (recall Remark \ref{lublin}) and hence
\[
P(\vert \tilde{Y_j}\vert> N)\leq \frac{E\vert \tilde{Y}_j\vert^{\tau}}{N^{\tau}}\leq
\frac{C}{N^{\tau}},
\]
for some fixed $C>0$. So for any $\eta>0$, $P(|\tilde{Y_j}| \in \mathbb{R} \setminus
[-\left({2C}/{\eta}\right)^{1/\tau},\left({2C}/{\eta}\right)^{1/\tau} ])< \eta$ for all $j$
hence  the sequence of the laws of $\tilde{Y}_j$ is tight. Then, by Lemma \ref{ratatouille},
the sequence of laws of
\[
Y_j:=(\varepsilon',\theta_1(j),\ldots,\theta_T(j),Z_1,\ldots,Z_T),
\]
is also tight and hence admits
a subsequence (which we continue to denote by $j$) weakly convergent to some probability law $\mu$
on $\mathcal{B}(\mathbb{R}^{T(d+N)+1})$.

We will construct, inductively, $\theta^*_t$, $t=1,\ldots,T$ such that
$(\varepsilon',\theta^*_1,\ldots,\theta^*_T,$ $Z_1,\ldots,Z_T)$ has law $\mu$ and $\theta^*$
is $\mathcal{F}$-predictable.
Let $M$ be a $T(d+N)+1$-dimensional random variable with law $\mu$. \\
First note that
$(M^{1+ Td+1}, \ldots,M^{1+Td+N})$ has the same law as $Z_1$,\\
$\ldots,$
$(M^{1+ Td+(T-1)N+1}, \ldots,M^{1+Td+TN})$ has the same law as $Z_T$.
% then coordinate $M^1$ will have same law than $\varepsilon'$,
%coordinates $M^2, \ldots,M^{1+d}$  than $\theta^*_1$, ...,
%coordinates $M^{1+(T-1)d}, \ldots,M^{1+Td}$  than $\theta^*_T$,
%coordinates $M^{1+ Td}, \ldots,M^{1+Td+N}$ than $Z_1$, ...,
%coordinates $M^{1+ Td+(T-1)N}, \ldots,M^{1+Td+TN}$ than $Z_T$.

Now let $\mu_k$ be the law of $(M^1,\ldots,M^{1+kd},M^{1+dT+1},\ldots,M^{1+dT+NT})$ on $\mathbb{R}^{kd+NT+1}$
(which represents the marginal of $\mu$ with respect to its
first $1+kd$ and last $NT$ coordinates), $k\geq 0$.

As a first step, we
apply Lemma \ref{ode} with $E:=\tilde{\varepsilon}$, $l:=T$
to get $\sigma(\tilde{\varepsilon})$-measurable random variables
$\varepsilon_1,\ldots,\varepsilon_T$ that are independent with uniform law
on $[0,1]$.

% (this will correspond respectively
%to the law of $(\varepsilon',\theta^*_1,\ldots,\theta^*_k,Z_1,\ldots,Z_T)$).

Applying Lemma \ref{preserve} with the choice
$N_1=d$, $N_2=1$, $Y=\varepsilon'$ and $E=\varepsilon_1$ we get a function $G$
such that $(\varepsilon',G(\varepsilon',\varepsilon_1))$ has the same law as the
marginal of $\mu_1$ with respect to its first $1+d$ coordinates.

We recall the following simple fact.
Let $Q$, $Q'$, $U$, $U'$ random variables such that $Q$ and $Q'$ have same law and
$U$ and $U'$ have same law. If $Q$ is independent of $U$ and $Q'$ is independent of $U'$, then
$(Q,U)$ and $(Q',U')$ have same law.

Let $Q=(M^1,\ldots,M^{d+1})$, $Q'=(\varepsilon', G(\varepsilon',\varepsilon_1))$,
$U=(M^{1+dT+1},\ldots,M^{1+dT+dN})$ and $U'=(Z_1,\ldots,Z_T)$.
As $(\varepsilon_1,\varepsilon',Z_1,\ldots,Z_T)$ are independent, we get that
$Q'$ is independent of $U'$.
Now remark that weak convergence preserves independence: since
$(\varepsilon',\theta_1(j))$ and $U'$ are independent for all $j$,
we get that $Q$ is independent of $U$. So we conclude that
$(\varepsilon', G(\varepsilon',\varepsilon_1),Z_1,\ldots,Z_T)$ has
law $\mu_1$. Define
$\theta_1^*:=G(\varepsilon',\varepsilon_1)$, this is clearly $\mathcal{F}_0$-measurable.

Carrying on, let us assume that we have found $\theta_j^*$, $j=1,\ldots,k$ such that
$(\varepsilon',\theta_1^*,\ldots,\theta_k^*,Z_1,\ldots Z_T)$ has law $\mu_k$ and
$\theta_j^*$ is a function of $\varepsilon',Z_1,\ldots,Z_{j-1},$ $\varepsilon_1,
\ldots,\varepsilon_j$ only (and is thus $\mathcal{F}_{j-1}$-measurable). We apply Lemma
\ref{preserve} with $N_1=d$, $N_2=kd+kN+1$,
$E=\varepsilon_{k+1}$ and
\[
Y=(\varepsilon',\theta_1^*,\ldots,\theta_k^*,Z_1,\ldots,Z_k)
\]
to get $G$ such that $ (Y,G(Y,\varepsilon_{k+1})) $ has the same law
as $(M^1,\ldots,M^{1+kd},M^{1+Td+1},$
$\ldots,M^{1+Td+kN},$ $M^{1+kd+1},\ldots,M^{1+(k+1)d})$. Thus
\[
Q'=(\varepsilon',\theta_1^*,\ldots,\theta_k^*,G(Y,\varepsilon_{k+1}),Z_1,\ldots,Z_k)
\]
has the same law as $Q=(M^1,\ldots,M^{1+(k+1)d},M^{1+Td+1},\ldots,M^{1+Td+kN})$,
the marginal of $\mu_{k+1}$ with respect to its
first $1+Td+kN$ coordinates.
%and with respect to its coordinates
%ranging from $1+(k+1)d+1$ to $1+(k+1)d+kN$ (this will correspond respectively
%to the law of $(\varepsilon',\theta^*_1,\ldots,\theta^*_{k+1},Z_1,\ldots,Z_k)$).
Now choose
$U=(M^{1+dT+kN+1},$
$\ldots,M^{1+dT+dN})$ (the marginal of $\mu_{k+1}$ with respect to its
 $(T-k)N$ last remaining   coordinates) and  $U'=(Z_{k+1},\ldots,Z_T)$.
As $Q'$ depends only on
$(\varepsilon_1,\ldots,\varepsilon_{k+1},\varepsilon',$ $Z_1,\ldots,Z_{k})$, which is independent
from
$(Z_{k+1},\ldots,Z_T)$, $Q'$
is independent of $U'$. Moreover,
$(\varepsilon',\theta_1(j),\ldots,\theta_{k+1}(j),$ $Z_{1},\ldots,Z_{k})$ and $(Z_{k+1},\ldots,Z_T)$
are independent for all $j$ and weak convergence preserves independence,
so $Q$ is independent of
$U$.
This entails that
\[
(\varepsilon',\theta_1^*,\ldots,\theta_k^*,G(Y,\varepsilon_{k+1}),Z_1,\ldots,Z_{T})
\]
has law $\mu_{k+1}$ and setting $\theta_{k+1}^*:=G(Y,\varepsilon_{k+1})$
we make sure that $\theta_{k+1}^*$ is a function of $\varepsilon',Z_1,\ldots,Z_{k},$ $\varepsilon_1,
\ldots,\varepsilon_{k+1}$ only, a fortiori, it is $\mathcal{F}_k$-measurable.
We finally get all the $\theta_j^*$, $j=1,\ldots,T$ such that the law of
\[
(\varepsilon',\theta_1^*,\ldots,\theta_T^*,Z_1,\ldots,Z_T)
\]
equals $\mu=\mu_T$. We will now show that
\begin{equation}\label{fattou}
V(X_0;\theta^*_1,\ldots,\theta^*_T)\geq {\limsup}_{j\to\infty} V(X_0;\theta_1(j),\ldots,\theta_T(j)),
\end{equation}
which will conclude the proof.

Indeed, $H_j:=X_0+\sum_{t=1}^T \theta_t(j)\Delta S_t-B$ clearly converges in law
to $H:=X_0+\sum_{t=1}^T \theta^*_t\Delta S_t-B$, $j\to\infty$ (note that $\Delta S_t$
and $B$ are continuous functions of the $Z_t$ and $X_0$ is deterministic). By continuity of $u_+,u_-$ also
$u_{\pm}([H_j]_{\pm})$ tends to $u_{\pm}([H]_{\pm})$ in law which entails
that $P(u_{\pm}([H_j]_{\pm})\geq y)\to P(u_{\pm}([H]_{\pm})\geq y)$ for all $y$
outside a countable set (the points of discontinuities of the
cumulative distribution functions of $u_{\pm}([H]_{\pm})$).

It suffices thus to find a measurable function $h(y)$ with $w_+(P(u_{+}[H_j]_+\geq y))\leq h(y),
j\geq 1$ and
$\int_{0}^{\infty} h(y)dy<\infty$ and then {(sup)} Fatou's lemma will imply
\eqref{fattou}. We get, just like in Lemma \ref{eel}, using Chebishev's inequality, \eqref{tenet} and \eqref{fuga},
for $y\geq 1$:
\begin{eqnarray*}
w_+(P(u_{+}[H_j]_+\geq y)) & \leq &
%g_+  2^{\lambda -1}\frac{k_+^{\lambda} \left(E( \vert X_0 \vert^{\lambda \alpha_+ })+
%\sum_{t=1}^T E\left(\vert \theta_t(j)-\phi_t \vert^{\lambda \alpha_+ }
%E(\vert\Delta S_t\vert^{\lambda \alpha_+ }| \mathcal{F}_{t-1})\right)
%+|b|^{\lambda \alpha_+ }\right)+{k}_+^{\lambda} +1 }{y^{\lambda \gamma_+}}\\
%& \leq &
C \frac{1+\vert X_0 \vert^{\lambda \alpha_+ }+\sum_{t=1}^T
E\left(\vert \theta_t(j)-\phi_t \vert^{\lambda \alpha_+ }\vert\Delta S_t\vert^{\lambda \alpha_+ }\right)}{y^{\lambda \gamma_+}}\\
& \leq & \frac{C}{y^{\lambda\gamma_+}}\left(1+|X_0|^{\lambda \alpha_+ }+
\sum_{t=1}^T E^{1/p}|\theta_t(j)-\phi_t|^{\tau}E^{1/q}W_t^q\right),
\end{eqnarray*}
for some constant $C>0$ and $W_t\in\mathcal{W}^+$, $t=1,\ldots,T$,
using H\"older's inequality with $p:=\tau/(\lambda \alpha_+ )$ and
its conjugate $q$ (recall that $\Delta S_t \in\mathcal{W}_t$). We know
from the construction that $\sup_{j,t}
E\vert\theta_t(j)-\phi_t\vert^{\tau}<\infty$.
% and by Assumption
%\ref{as4} $E( \vert X_0 \vert^{\lambda \alpha_+ }) < \infty$.
Thus we can find some constant $C'>0$ such that
$w_+(P(u_{+}[H_j]_+\geq y)) \leq C'/y^{\lambda\gamma_+}$, for all
$j$. Now trivially $w_+(P(u_{+}[H_j]_+\geq y))\leq w_+(1)=1$ for
$0\leq y\leq 1$. Setting $h(y):=1$ for $0\leq y\leq 1$ and
$h(y):=C'/y^{\lambda\gamma_+}$ for $y>1$, we conclude since
$\lambda\gamma_+>1$ and thus $1/y^{\lambda\gamma_+}$ is integrable
on $[1,\infty)$.
\end{proof}

%Now take $\varepsilon_t$, $t=0,\ldots, T-1$ independent uniform on $[0,1]$ and
%$\sigma(\varepsilon)$-measurable.

%We claim that there is $\theta^*\in\mathcal{A}_0(X_0)$ such that the law of
%$(X_0,\theta_1^*,\ldots,\theta_T^*,\Delta S_1,\ldots,\Delta S_T)$ equals $\mu$.
%Indeed, set $\xi_0:=(X_0,\Delta S_1)$. Apply Lemma \ref{kitudja} recursively to $\epsilon_t$, $\xi_t:=(X_0,\theta_1^*,\ldots,\theta_t^*,\Delta S_1,\ldots,\Delta S_{t+1})$ and $\nu_t$, $t=0,\ldots, T-1$
%to get $\eta_t$, here $\nu_t$ denotes the marginal law of $\mu$ with respect to the coordinates that
%correspond to $(X_0,\theta_1^*,\ldots,\theta_t^*,\theta_{t+1}^*,\Delta S_1,\ldots,\Delta S_{t+1})$ and setting $\theta_{t+1}^*:=\eta_t$ we may conclude.

\section{Existence without using relaxed strategies}\label{secondrevision}

In the previous section, a class of ``relaxed'' strategies was considered in the
sense that the investor was allowed to make use of an external random source (i.e.
a random number generated by a computer), see Assumptions \ref{as3} and \ref{masikk}
above. One may wonder whether it is possible to prove the existence of an optimal
strategy in the class of non-relaxed strategies. We will see that this is
possible under suitable hypotheses.

Assumption \ref{AAA} below states that the filtration is
generated by the independent random shocks $Z_t$, $t\geq 1$ which move the prices
at $t$ and the filtration is rich enough in information in the sense that there are risks arising at time $t$ in the market
(represented by $U_t$) that are not hedgeable by the traded
financial instruments. In other words, there is enough ``noise'' in the market
(which is the case in most real markets).

Assumption \ref{AAA} is satisfied
in a broad class of processes that are
natural discretizations of continuous-time diffusion models for
asset prices. It holds, roughly
speaking, when the market is ``incomplete'', see the examples of section \ref{exaoa} below for more details.

\begin{assumption}\label{AAA}
Let $\mathcal{F}_0=\{\emptyset,\Omega\}$, $\mathcal{F}_t=\sigma(Z_1,\ldots,Z_t)$ for $t=1,
\ldots,T$,
where the $Z_i$, $i=1,\ldots, T$ are $\mathbb{R}^N$-valued
independent random variables. $S_0$ is constant
and $S_1=f_1(Z_1)$, $S_t=f_t(S_1,\ldots,S_{t-1},Z_t)$, $t=2,\ldots,T$ for some continuous functions $f_t$ (hence $S_t$ is adapted).

Furthermore, for $t=1,\ldots,T$ there exists an $\mathcal{F}_t$-measurable
uniformly distributed random variable $U_t$ which is independent of $\mathcal{F}_{t-1}\vee \sigma (S_t)$.
\end{assumption}

Assumption \ref{as4} needs to be replaced by
\begin{assumption}\label{as44}
The random variable $B$ is a continuous function of $(S_1,\ldots,S_{T})$ and
$\mathcal{A}(X_0)$ is not empty. $u_{\pm},w_{\pm}$ are continuous functions.
\end{assumption}

\begin{remark} {\rm Note that under Assumption \ref{AAA}, the initial
capital $X_0$ is necessarily
deterministic.}
\end{remark}

The main result of the present section is the following.

\begin{theorem}\label{maine-new} Let Assumptions \ref{marche}, \ref{as1}, \ref{as2}, \ref{AAA} and \ref{as44}
hold. Then there is $\theta^*\in\mathcal{A}(X_0)$ such that
\[
V(X_0;\theta^*_1,\ldots,\theta^*_T)=\sup_{\theta\in\mathcal{A}(X_0)}V(X_0;\theta_1,\ldots,\theta_T)<\infty.
\]
\end{theorem}

The proof is similar to that of Theorem \ref{maine}.
\begin{proof}[Proof of Theorem \ref{maine-new}]
As in the proof of Theorem \ref{maine},
take $\theta(j)\in\mathcal{A}(X_0),j\in\mathbb{N}$ such that
\[
V(X_0;\theta_1(j),\ldots,\theta_T(j))\to \sup_{\theta\in\mathcal{A}(X_0)}V(X_0;\theta_1,\ldots,\theta_T),
\quad j\to\infty.
\]
%and, by Lemma \ref{ujdonsag}, we may assume $\theta(j)\in\mathcal{A}'(X_0)$
%for all $j$.
Using {Theorem \ref{ajtothm}, Lemmata \ref{eel},} \ref{porto} and \ref{ratatouille} just
like in the proof of Theorem \ref{maine}
we find that a subsequence {in $\mathcal{A}(X_0)$} (still denoted by $j$) of the $2Td$-dimensional
random variables
\[
Y_j:=(S_1,\ldots,S_T,\theta_1(j),\ldots,\theta_T(j)),
\]
converges weakly to some probability law $\mu$
on $\mathcal{B}(\mathbb{R}^{2Td})$.
We will also use the notation $Y^{(k)}_j:=(S_1,\ldots,S_k,\theta_1(j),\ldots,\theta_k(j))$
and denote its law on $\mathcal{B}(\mathbb{R}^{2kd})$ by ${\mu}_k(j)$.
Let $M$ be a $2Td$-dimensional random variable with law $\mu$. Let $\mu_k$ be the law of
$(M^1,\ldots,M^{kd},M^{Td+1},\ldots,M^{Td+kd})$ on $\mathcal{B}(\mathbb{R}^{2kd})$.
Note that the law of $Y^{(k)}_j$, ${\mu}_k(j)$, weakly converges to $\mu_k$.

We shall construct, inductively, $\theta^*_i$, $i=1,\ldots,T$ such that
$F_k:=(S_1,\ldots,S_{k},\theta^*_1,\ldots,\theta^*_k)$ has law $\mu_k$ for
all $k=1,\ldots,T$, and $\theta^*=(\theta^*_1,\ldots,\theta^*_T)$
is $\mathcal{F}$-predictable.

As $\theta_1(j)$ are deterministic numbers, weak convergence implies that
they converge to some (deterministic) $\theta^*_1$ which is then $\mathcal{F}_0$-measurable.
Clearly, $(S_1,\theta_1^*)$ has law $\mu_1$.

Carrying on, let us assume that we have found $\theta_i^*$, $i=1,\ldots,k$ such that
$F_k$
has law $\mu_k$ and
$\theta_j^*$ is $\mathcal{F}_{j-1}$-measurable for $j=1,\ldots,k$.

We now apply Lemma
\ref{preserve} with $N_1=d$, $N_2=2kd$,
$E=U_{k}$ and $Y=F_k$
to get $G$ such that $(F_k,G(F_k,U_{k})) $ has the same law
as $(M^1,\ldots,M^{kd},M^{Td+1},M^{Td+(k+1)d})$, we denote this law by $\bar{\mu}_k$
henceforth (note that, by Assumption
\ref{AAA}, $U_{k}$ is independent of {$F_k$}).
Define $\theta^*_{k+1}:=G(F_k,U_{k})$, this is clearly $\mathcal{F}_k$-measurable.
It remains to show that $F_{k+1}$ has law $\mu_{k+1}$. As
$\mu_{k+1}$ is the weak limit of $Y^{(k+1)}_j$, it is enough to
 prove that  the weak limit of $Y^{(k+1)}_j$ is $F_{k+1}$. We first express
the laws of $Y^{(k+1)}_j$ and $F_{k+1}$ by mean of conditioning.

By Assumption \ref{AAA} one can write
$S_{k+1}=f_{k+1}(S_1,\ldots,S_k,Z_{k+1})$ with some continuous
function $f_{k+1}$. Notice that the law of the
$2(k+1)d$-dimensional random variable $F_{k+1}$ is $$
\mu_{k+1}(dx)=
\bar{\mu}_k(d\sigma_1,\ldots,d\sigma_{k},d\tau_{1},\ldots,
d\tau_{k+1})
\rho(d\sigma_{k+1}|\sigma_1,\ldots,\sigma_{k},\tau_{1},\ldots,
\tau_{k+1})
$$
where we write $dx=(dx_1,\ldots,dx_{2(k+1)d})=(d\sigma_1,\ldots,d\sigma_{k+1},d\tau_{1},\ldots,
d\tau_{k+1})$,
$d\sigma_{j}=(dx_{(j-1)d+1},\ldots,dx_{jd})$ for $j=1,\ldots,k+1$ and $d\tau_{i}=(dx_{(k+i)d+1},\ldots,dx_{(k+i+1)d})$
for $i=1,\ldots,k+1$. The probabilistic kernel $\rho$ is defined by
\begin{eqnarray*}
\rho(A|\sigma_1,\ldots,\sigma_{k},\tau_{1},\ldots, \tau_{k+1})&:= &
P(S_{k+1} \in A |S_1=\sigma_1,\ldots,S_k=\sigma_{k},\theta^*_1=\tau_{1},\ldots,\theta^*_{k+1}= \tau_{k+1})\\
& = & P(f_{k+1}(\sigma_1,\ldots,\sigma_{k},Z_{k+1})\in A|S_1=\sigma_1,\ldots,S_k=\sigma_{k})\\
& = & P(f_{k+1}(\sigma_1,\ldots,\sigma_{k},Z_{k+1})\in A),
\end{eqnarray*}
for $A\in\mathcal{B}(\mathbb{R}^d)$, $(\sigma_1,\ldots,\sigma_{k}) \in \mathbb{R}^{kd}$  and
$(\tau_{1},\ldots, \tau_{k+1}) \in \mathbb{R}^{(k+1)d}$.
The crucial observation here is that $\rho$ does not depend on $(\tau_{1},\ldots, \tau_{k+1})$.

It follows in the same way that, for all $j$,
the law of $Y^{(k+1)}_j$
is $$\mu_{k+1}(j)(dx)=
\bar{\mu}_{k}(j)(d\sigma_1,\ldots,d\sigma_{k},d\tau_{1},\ldots, d\tau_{k+1})
\rho(d\sigma_{k+1}|\sigma_1,\ldots,\sigma_{k},\tau_{1},\ldots, \tau_{k+1})
,$$
where $\bar{\mu}_{k}(j)$ is the law of $(S_1,\ldots,S_k,\theta_1(j),\ldots,\theta_{k+1}(j))$.

Clearly, the weak convergence of the $\mathrm{Law}(Y_j)$ to $\mu$ implies that
their marginals $\bar{\mu}_k(j)$ converge weakly to $\bar{\mu}_k$, for each $k$.
To conclude the proof, we have to show that this implies
also
\begin{eqnarray}
\nonumber
\mu_{k+1}(j)(dx)=
\bar{\mu}_{k}(j)(d\sigma_1,\ldots,d\sigma_{k},d\tau_{1},\ldots, d\tau_{k+1})
\rho(d\sigma_{k+1}|\sigma_1,\ldots,\sigma_{k},\tau_{1},\ldots, \tau_{k+1})  \to
\\
\label{convcond}
\mu_{k+1}(dx)=
\bar{\mu}_k(d\sigma_1,\ldots,d\sigma_{k},d\tau_{1},\ldots,
d\tau_{k+1})
\rho(d\sigma_{k+1}|\sigma_1,\ldots,\sigma_{k},\tau_{1},\ldots,
\tau_{k+1})
\end{eqnarray}
weakly as $j\to\infty$.

First notice that, for any sequence $z^n\to z$ in $\mathbb{R}^{(2k+1)d}$,
$\rho(\cdot|z^n)$ tends to $\rho(\cdot |z)$ weakly. Indeed, taking any
continuous and bounded $h$ on $\mathbb{R}^d$, we have
\begin{eqnarray*}
\int_{\mathbb{R}^d}h(\sigma)\rho(d\sigma|z^n) &=& Eh(f_{k+1}(z^n_1,\ldots,z^n_{kd},Z_{k+1})) \to
\\
Eh(f_{k+1}(z_1,\ldots,z_{kd},Z_{k+1})) &=& \int_{\mathbb{R}^d}h(\sigma)\rho(d\sigma|z)
\end{eqnarray*}
by continuity of $h,f_{k+1}$, boundedness of $h$ and Lebesgue's theorem.

Now take any uniformly continuous and bounded $g:\mathbb{R}^{2(k+1)d}\to\mathbb{R}$.
Define
$$
\bar{g}(z):=\int_{\mathbb{R}^d}g(z,\sigma)\rho(d\sigma|z),\quad z\in\mathbb{R}^{(2k+1)d}.
$$
We claim that $\bar{g}$ is continuous. Indeed, let $z^n\to z$. Then
\begin{eqnarray*}
|\bar{g}(z^n)-\bar{g}(z)|\leq  |\int_{\mathbb{R}^d}g(z^n,\sigma)\rho(d\sigma|z^n)-
\int_{\mathbb{R}^d}g(z,\sigma)\rho(d\sigma|z^n)|+|\int_{\mathbb{R}^d}g(z,\sigma)\rho(d\sigma|z^n)-
\int_{\mathbb{R}^d}g(z,\sigma)\rho(d\sigma|z)|.
\end{eqnarray*}
Here the first term tends to zero by uniform continuity, the second term tends to
zero by the weak convergence of $\rho(\cdot|z^n)$ to $\rho(\cdot|z)$. This shows
the continuity of $\bar{g}$.

As $\bar{\mu}_k(j)$ converge weakly to $\bar{\mu}_k$, it follows that
\begin{eqnarray*}
\int_{\mathbb{R}^{(2k+1)d}}\bar{g}(\sigma_1,\ldots,\sigma_{k},\tau_{1},\ldots,
\tau_{k+1})\bar{\mu}_k(j)(d\sigma_1,\ldots,d\sigma_{k},d\tau_{1},\ldots,
d\tau_{k+1}) \to \\
\int_{\mathbb{R}^{(2k+1)d}}\bar{g}(\sigma_1,\ldots,\sigma_{k},\tau_{1},\ldots,
\tau_{k+1})\bar{\mu}_k(d\sigma_1,\ldots,d\sigma_{k},d\tau_{1},\ldots,
d\tau_{k+1}) \end{eqnarray*}
This implies that
\begin{eqnarray}
\nonumber
\int_{\mathbb{R}^{(2k+2)d}}g(\sigma_1,\ldots,\sigma_{k},\tau_{1},\ldots,
\tau_{k+1},\sigma_{k+1})\rho(d\sigma_{k+1}|\sigma_1,\ldots,\sigma_{k},\tau_{1},\ldots, \tau_{k+1})
\bar\mu_k(j)(d\sigma_1,\ldots,d\sigma_{k},d\tau_{1},\ldots,
d\tau_{k+1})
\to \\
\label{lippa}
\int_{\mathbb{R}^{(2k+2)d}}g(\sigma_1,\ldots,\sigma_{k},\tau_{1},\ldots,
\tau_{k+1},\sigma_{k+1})
\rho(d\sigma_{k+1}|\sigma_1,\ldots,\sigma_{k},\tau_{1},\ldots, \tau_{k+1})
\bar\mu_k(d\sigma_1,\ldots,d\sigma_{k},d\tau_{1},\ldots,
d\tau_{k+1}),
\end{eqnarray}
showing that \eqref{convcond} holds
%$\bar{\mu}_k(ds_1,\ldots,ds_{kd},ds_{(k+1)d+1},ds_{2(k+1)d})
%\rho(dz|s_1,\ldots,s_{kd},s_{(k+1)d+1},s_{2(k+1)d})$
%tends to $\bar{\mu}_k(j)(ds_1,\ldots,ds_{kd},ds_{(k+1)d+1},ds_{2(k+1)d})
%\rho(dz|s_1,\ldots,s_{kd},s_{(k+1)d+1},s_{2(k+1)d})$
%weakly as $j\to\infty$
(recall that, in order to check weak convergence, it is enough
to verify \eqref{lippa} for uniformly continuous bounded functions, see Theorem 1.1.1 of
\cite{stroock-varadhan}) and the induction step is completed.

%This implies that the law of
%$$(M^1,\ldots,M^{(k+1)d},M^{Td+1},\ldots,
%M^{Td+(k+1)d}),$$
%which is the weak limit of $\bar{\mu}_k(j)(ds_1,\ldots,ds_{kd},ds_{(k+1)d+1},ds_{2(k+1)d})
%\rho(dz|s_1,\ldots,s_{kd},s_{(k+1)d+1},s_{2(k+1)d})$ by construction,
%is the same as the law of
%$$(S_1,\ldots,S_{k+1},\theta^*_1,\ldots,\theta^*_{k+1}),
%$$
%which is $\bar{\mu}_k(ds_1,\ldots,ds_{kd},ds_{(k+1)d+1},ds_{2(k+1)d})
%\rho(dz|s_1,\ldots,s_{kd},s_{(k+1)d+1},s_{2(k+1)d})$
%as pointed out above. We notice that the law of $(M^1,\ldots,M^{(k+1)d},M^{Td+1},\ldots,
%M^{Td+(k+1)d})$ is $\mu_{k+1}$ by definition, so the induction step is completed.

We finally arrive at $(S_1,\ldots,S_T,\theta^*_1,\ldots,\theta^*_T)$
with law $\mu_T=\mu$. We can show  verbatim as in the proof of Theorem
\ref{maine} that
\begin{equation}
V(X_0;\theta^*_1,\ldots,\theta^*_T)\geq\limsup_{j\to\infty} V(X_0;\theta_1(j),\ldots,\theta_T(j)),
\end{equation}
using the properties of weak convergence and that $B$ is a continuous function of the
$S_t$, $t=1,\ldots,T$. This concludes the proof.
\end{proof}

\section{Examples}
\label{exaoa}
In this section, we first present some classical market models where Assumptions \ref{marche} and \ref{as3} hold true and hence Theorem \ref{maine} applies.
\begin{example}
{\rm Let $S_0$ be constant and $\Delta S_t\in\mathcal{W}$ independent $t=1,\ldots,T$.
Take $Z_i:=\Delta S_i$, define $\mathcal{G}_0:=\{\emptyset,\Omega\}$
and $\mathcal{G}_t:=\sigma(Z_1,\ldots,Z_t)$, $T\geq 1$.
Assume that $S_t$ satisfies (NA) + (R) w.r.t. $\mathcal{G}_t$. Then this continues to hold
for the enlargement $\mathcal{F}_t$ defined in Assumption \ref{as3}. So
Assumptions \ref{marche} and \ref{as3} hold with $\kappa_t,\pi_t$ almost surely
constants since the conditional law of $\Delta S_t$ w.r.t.
$\mathcal{F}_{t-1}$ is a.s. equal to its actual law.}
\end{example}

\begin{example}\label{bhj}
{\rm Fix $d\leq L\leq N$. Take $Y_0\in\mathbb{R}^L$ constant and
define $Y_t$ by the difference equation
\[
Y_{t+1}-Y_t=\mu(Y_t)+\nu(Y_t)Z_{t+1},
\]
where $\mu:\mathbb{R}^L\to\mathbb{R}^L$ and
$\nu:\mathbb{R}^L\to\mathbb{R}^{L\times N}$ are bounded and
continuous. We assume that there is $h>0$ such that
\begin{equation}\label{revizor}
v^T\nu(x)\nu^T(x)v\geq hv^Tv,\quad v\in\mathbb{R}^L,
\end{equation}
for all $x\in\mathbb{R}^L$; $Z_t\in\mathcal{W}$, $t=1,\ldots,T$ are
independent with $\mathrm{supp}\,\mathrm{Law}\,Z_t=\mathbb{R}^N$.}

{\rm Thus $Y_t$ follows a discretized dynamics of a non-degenerate
diffusion process. We may think that $Y_t$ represent the evolution
of $L$ economic factors or, more specifically, of some assets. Take
$\mathcal{G}_0$ trivial and $\mathcal{G}_t:=\sigma(Z_j,j\leq t)$,
$t\geq 1$.}

{\rm We claim that $Y_t$ satisfies Assumption \ref{marche} with respect to
$\mathcal{G}_t$. Indeed, $Y_t\in\mathcal{W}$ is trivial and we will show that \eqref{rrr} holds
with $\kappa_t,\pi_t$ constants.}

{\rm Take $v\in\mathbb{R}^L$. Obviously,
$$
P(v(Y_{t+1}-Y_t)\leq -\vert v\vert | \mathcal{G}_t)=P(v(Y_{t+1}-Y_t)\leq -\vert v\vert |  Y_t).
$$
It is thus enough to show for each $t=1,\ldots,T$ that there is
$c>0$ such that for each unit vector $v$ and each $x\in\mathbb{R}^L$
$$
P(v(\mu(x)+
\nu(x)Z_t)\leq -1)\geq c.
$$
Denoting by $m$ an upper bound for $\vert\mu(x)\vert$,
$x\in\mathbb{R}^L$, we may write
$$
P(v(\mu(x)+
\nu(x)Z_t)\leq -1)\geq
P(v(\nu(x)Z_t)\leq -(m+1)).
$$
Here $y=v^T\nu(x)$ is a vector of length at least $\sqrt{h}$, hence
the absolute value of one of its components is at least
$\sqrt{h/N}$.
%We may assume that $y^1$ is this component and $y^1>0$.
Thus we have
\begin{eqnarray*}
P(v^T\nu(x)Z_t\leq -(m+1)) & \geq &
\min \left( \min_{i,k_i} P(\sqrt{h/N} Z_t^i \leq -(m+1),
k_i(j) Z_t^j\leq 0,\ j\neq i); \right. \\
 & &  \left. \min_{i,k_i} P(\sqrt{h/N}Z_t^i \geq (m+1),
k_i(j) Z_t^j\leq 0,\ j\neq i) \right)
\end{eqnarray*}
where $i$ ranges over $1,\ldots,N$ and $k_i$ ranges over the
(finite) set of all functions from $\{1,2,\ldots,i-1,i+1,\ldots,N\}$
to $\{1,-1\}$ (representing all the possible configurations for the
signs of $y^j$, $j\neq i$). This minimum is positive by our
assumption on the support of $Z_t$.}

{\rm Now we can take $S_t^i:=Y_t^i$, $i=1,\ldots,d$ for some $d\leq
L$. When $L>d$, we may think that the $Y_j$, $d<j\leq L$ are not
prices of some traded assets but other relevant economic variables
that influence the market. It is easy to check that Assumption
\ref{marche} holds for $S_t$, too, with respect to $\mathcal{G}_t$.}

{\rm Enlarging each $\mathcal{G}_t$ by ${\varepsilon}$,
independent of $Z_1,\ldots,Z_T$, we get $\mathcal{F}_t$
as in Assumption \ref{as3}. Clearly, Assumption \ref{marche} continues to hold for $S_t$
with respect to $\mathcal{F}_t$ and Assumption \ref{as3} is then also true as $S_t$ is
a continuous function of $Z_1,\ldots,Z_t$.}
\end{example}

\begin{example}\label{bhjj} {\rm Take $Y_t$ as in the above example. For simplicity, we assume
$d=L=N=1$ and $\nu(x)>0$ for all $x$. Furthermore, let $Z_t$,
$t=1,\ldots,T$ be such that for all $\zeta>0$,
\[
Ee^{\zeta\vert Z_t\vert}<\infty.
\]
Set $S_t:=\exp(Y_t)$ this time. We claim
that Assumption \ref{marche} holds true for $S_t$ with respect to the filtration
$\mathcal{G}_t$. Obviously, $\Delta S_t\in\mathcal{W}$, $t\geq 1$.}

{\rm We choose $\kappa_t:=S_t/2$. Clearly,
$1/\kappa_t\in\mathcal{W}$. It suffices to prove that $
1/P(S_{t+1}-S_t\leq -S_t/2 | \mathcal{G}_t)$ and
$1/P(S_{t+1}-S_t\geq S_t/2 | \mathcal{G}_t)$ belong to
$\mathcal{W}$. We will show only the second containment, the first
one being similar. This amounts to checking
\begin{equation*}\label{wiwi}
1/P(\exp\{Y_{t+1}-Y_t\}\geq 3/2 |  Y_t)\in\mathcal{W}.
\end{equation*}
Let us notice that
\begin{eqnarray*}
P(\exp\{Y_{t+1}-Y_t\}\geq 3/2 |  Y_t) & = & P(\mu(Y_t)+\nu(Y_t)Z_{t+1}\geq \ln (3/2) |  Y_t)\\
 & = & P(Z_{t+1}\geq \frac{\ln(3/2)-\mu(Y_t)}{\nu(Y_t)} |  Y_t)\\
 & \geq &
P(Z_{t+1}\geq \frac{\ln(3/2)+m}{\sqrt{h}}),
\end{eqnarray*}
which is a deterministic positive constant, by the assumption on the support of $Z_{t+1}$. Defining the
enlarged $\mathcal{F}_t$, Assumptions \ref{marche} and \ref{as3} hold for $S_t$.
Examples \ref{bhj} and \ref{bhjj} are pertinent, in particular, when the $Z_t$
are Gaussian.
}
\end{example}

We now show an example where Assumption \ref{AAA} holds and hence Theorem
\ref{maine-new} applies.

\begin{example}\label{conclusive}
{\rm Let us consider the same setting as in Example \ref{bhj} with
$d=L< N$. This corresponds to the case when an incomplete diffusion
market model has been discretized (the number of driving processes,
$N$, exceeds the number of assets, $d$). Let us furthermore assume
that for all $t$, the law of $Z_t$ has a density w.r.t. the
$N$-dimensional Lebesgue measure (when we say ``density'' from now
on we will always mean density w.r.t. a Lebesgue measure of
appropriate dimension). Recall that
$\mathcal{F}_0=\{\emptyset,\Omega\}$ and
$\mathcal{F}_t=\sigma(Z_1,\ldots,Z_t)$ for $t=1, \ldots,T$,}

{\rm It is clear that $S_{t+1}=f_{t+1}(S_1,\ldots,S_t,Z_{t+1})$ for
some continuous function $f_{t+1}$. It remains to construct
$U_{t+1}$ as required in Assumption \ref{AAA}.}

{\rm We will denote by $\nu_i(x)$ the $i$th row of $\nu(x)$,
$i=1,\ldots,d$. First let us notice that \eqref{revizor} implies
that $\nu(x)$ has full rank for all $x$ and hence the $\nu_i(x)$,
$i=1,\ldots,d$ are linearly independent for all $x$.}

{\rm It follows that the set $\{(\omega,w)\in\Omega\times
\mathbb{R}^N: \nu_i(Y_t)w=0,\ i=1,\ldots,d,\ |w|=1\}$ has full
projection on $\Omega$ and it is easily seen to be in
$\mathcal{F}_t\otimes \mathcal{B}(\mathbb{R}^N)$. It follows by
measurable selection (see e.g. Proposition III.44 of \cite{dm}) that
there is a $\mathcal{F}_t$-measurable $N$-dimensional random
variable $\xi_{d+1}$ such that $\xi_{d+1}$ has unit length and it is
a.s. orthogonal to $\nu_i(Y_t)$, $i=1,\ldots,d$. Continuing in a
similar way we get $\xi_{d+1},\ldots,\xi_N$ such that they have unit
length, they are a.s. orthogonal to each other as well as to the
$\nu_i(Y_t)$. Let $\Sigma$ denote the $\mathbb{R}^{N\times
N}$-valued $\mathcal{F}_t$-measurable random variable whose rows are
$\nu_1,\ldots,\nu_d,\xi_{d+1}, \ldots,\xi_N$. Note that $\Sigma$ is
a.s. nonsingular (by construction and by \eqref{revizor}).}

{\rm As $Z_{t+1}$ is independent of $\mathcal{F}_t$ and $\Sigma$ is
$\mathcal{F}_t$-measurable, for any $(z_1,\ldots,z_t) \in
\mathbb{R}^{tN}$, the conditional law of $\Sigma Z_{t+1}$ knowing
$\{Z_1=z_1,\ldots,Z_t=z_t\}$  equals the law of the random variable
$\Sigma(z_1,\ldots,z_t) Z_{t+1}$. Recall that $Z_{t+1}$ has a density w.r.t. the $N$-dimensional
Lebesgue measure and that $Z \to \Sigma(z_1,\ldots,z_t) Z$ is a
continuously differentiable diffeomorphism since
$\Sigma(z_1,\ldots,z_t) $ is nonsingular. So we can use the change
of variable theorem and we deduce that $\Sigma(z_1,\ldots,z_t) Z_{t+1}$, and thus a.s. the conditional
law of  $\Sigma
Z_{t+1}$ knowing $\mathcal{F}_t$, has a  density.

As $(\nu
(Y_t)Z_{t+1}, \xi_{d+1} Z_{t+1})$ is the first $d+1$
coordinates of $\Sigma Z_{t+1}$, using Fubini theorem, the conditional law of
$(\nu
(Y_t)Z_{t+1}, \xi_{d+1} Z_{t+1})$  knowing $\mathcal{F}_t$ also has a  density. Using again the change
of variable theorem, it follows that the random variable
$(Y_{t+1},\xi_{d+1} Z_{t+1})$ has a $\mathcal{F}_t$-conditional density.
This implies that $\xi_{d+1} Z_{t+1}$ has a
$\mathcal{F}_t\vee\sigma(Y_{t+1})$-conditional  density and, a
fortiori, its conditional law is atomless.}

{\rm Lemma \ref{kella} with the choice $X:=\xi_{d+1} Z_{t+1}$ and
$W:=(Z_1,\ldots,Z_t,Y_{t+1})$ provides a uniform $U_{t+1}=G(\xi_{d+1} Z_{t+1},
Z_1,\ldots,Z_t,Y_{t+1})$ independent of $\sigma(Z_1,\ldots,Z_t,Y_{t+1})=\mathcal{F}_t\vee\sigma(Y_t)$
but $\mathcal{F}_{t+1}$-measurable (since $\xi_{d+1}Z_{t+1}$ is
$\mathcal{F}_{t+1}$-measurable and $G$ is measurable from Lemma \ref{kella}).
It follows that this example satisfies Assumption \ref{AAA} and hence Theorem \ref{maine-new} applies to it.}
\end{example}

\begin{remark}
{\rm Clearly, Assumption \ref{AAA} permits a non-Markovian price process
$S_{t}$ as well (i.e. $S_t$ may well depend on its whole past $S_{t-1},\ldots,S_1$).
Also, $S_{t}$ may be a
non-linear function of $S_1,\ldots,S_{t-1},Z_{t}$ in a more complex way
than in Example \ref{conclusive}. It is, however, outside the scope of
the present paper to go into more details here.}
\end{remark}

\section{Appendix}
\label{appendix}
\subsection{Proofs of Lemmas \ref{eel}, \ref{crux}, \ref{oslo}, \ref{porto} and of Proposition \ref{kh}}
\subsubsection{Proof of Lemma \ref{eel}}
\label{appendix eel}
We get, using \eqref{fuga} and Chebishev's inequality:
\begin{eqnarray}
V^+(X_{0};\theta_{1},\ldots,\theta_T)\leq
1+g_+ \int_1^{\infty}  \frac{
E^{\gamma_+}\left(u_+^{\lambda}([X_0+\sum_{n=1}^T \theta_n\Delta S_n-B]_+)\right)}{y^{\lambda\gamma_+}}dy
\end{eqnarray}
Evaluating the integral and using \eqref{tenet} we continue the estimation
as
\begin{eqnarray*}
V^+(X_{0};\theta_{1},\ldots,\theta_T)  & \leq &
1+\frac{g_+ }{\lambda\gamma_+-1}E^{\gamma_+}\left(
2^{\lambda-1}k_+^{\lambda} [X_0+\sum_{n=1}^T \theta_n\Delta S_n-B]_+^{\lambda \alpha_+ }
+2^{\lambda-1}{k}_+^{\lambda}\right) \\
 & \leq & 1+\frac{g_+ }{\lambda\gamma_+-1}
\left[
2^{\lambda-1} k_+^{\lambda} \left(E( \vert X_0+\sum_{n=1}^T (\theta_n-\phi_n)\Delta S_n\vert^{\lambda \alpha_+ })
+|b|^{\lambda \alpha_+ }\right) \right. \\
& & \left. + 2^{\lambda-1}{k}_+^{\lambda} +1 \right],
\end{eqnarray*}
using the rough estimate $x^{\gamma_+}\leq x+1$, $x\geq 0$,
Assumption \ref{as2} and the fact that $C_1\geq C_2$ implies that
$(Y-C_1)_+\leq (Y-C_2)_+$. This gives the first statement. For the
second inequality note that, by {\eqref{tenet-}}, \eqref{b} and Assumption \ref{as2},
\begin{eqnarray*}
V^-(X_0;\theta_{1},\ldots,\theta_{T}) & \geq &
g_-\Int_0^{\infty} P\left(u_-
([X_T^{X_0,\theta}-B]_-)\geq y\right)dy
 =  g_-Eu_-[X_T^{X_0,\theta}-B]_-\\
& \geq & g_- k_-\left(E[X_T^{X_0,\theta}-B]_-^{\alpha_-} -1\right)\\
& \geq & g_- k_- \left(E[X_0+\sum_{n=1}^{T} (\theta_n-\phi_n)\Delta S_n-b]_-^{\alpha_-}
-1\right).
\end{eqnarray*}

\subsubsection{Proof of Lemma \ref{crux}}
\label{appendix crux}
Notice that for $t=T$ the statement of Lemma \ref{crux} is trivial as there
are no strategies involved. Let us assume that Lemma \ref{crux} is true for
$t+1$, we will deduce that it holds true for $t$, too. Let $X_t\in \Xi_t^1$
%with $E\vert X_t\vert^{\alpha_-}<\infty$
and $(\theta_{t+1},\ldots,\theta_T)\in\tilde{\mathcal{A}}_t(X_t)$. Let
$X_{t+1}:=X_t+(\theta_{t+1}-\phi_{t+1})\Delta S_{t+1}$, then $X_{t+1}\in \Xi_{t+1}^1$
% with $E\vert X_{t+1}\vert^{\alpha_-}<\infty$ (je ne vois pas pourquoi).
and by Lemma \ref{tavasz}, $(\theta_{t+2},\ldots,\theta_T)\in\tilde{\mathcal{A}}_{t+1}(X_{t+1})$.
By
induction hypothesis, there exists $C_n^{t+1}$, $n=t+1,\ldots,T$ and
$(\hat{\theta}_{t+2},\ldots,\hat{\theta}_T)\in\tilde{\mathcal{A}}_{t+1}(X_{t+1})$ satisfying
\begin{equation}\label{robinhood}
\vert\hat{\theta}_n-\phi_n\vert\leq C^{t+1}_{n-1}[\vert X_t+
(\theta_{t+1}-\phi_{t+1})\Delta S_{t+1}\vert+1],
\end{equation}
and
$
\tilde{V}_{t+1}(X_{t+1};\theta_{t+2},\ldots,\theta_T)\leq
\tilde{V}_{t+1}(X_{t+1};\hat{\theta}_{t+2},\ldots,\hat{\theta}_T).$
It is clear from \eqref{robinhood} that
\[
\left\vert \sum_{n=t+2}^T (\hat{\theta}_n-\phi_n)\Delta S_n\right\vert\leq H\left(\vert X_t+(\theta_{t+1}-
\phi_{t+1})\Delta S_{t+1}\vert+1\right)
\]
for $H=\sum_{n=t+2}^TC_{n-1}^{t+1}\vert \Delta S_n \vert\in \mathcal{W}^+$. We have
\begin{eqnarray}
\nonumber
\tilde{V}_t(X_t;\theta_{t+1},\ldots,\theta_T) & = & E(\tilde{V}_{t+1}(X_{t}+(\theta_{t+1}-
\phi_{t+1})\Delta S_{t+1};\theta_{t+2},\ldots,\theta_T) | \mathcal{F}_t)\\
\nonumber
 & \leq &  E(\tilde{V}_{t+1}(X_t+(\theta_{t+1}-
\phi_{t+1})\Delta S_{t+1};\hat{\theta}_{t+2},\ldots,\hat{\theta}_T)| \mathcal{F}_t)\\
\label{ineqvtilde}
 & = & \tilde{V}_{t}(X_t;{\theta}_{t+1},\hat{\theta}_{t+2},\ldots,\hat{\theta}_T).
\end{eqnarray}
Fix some $\lambda \alpha_+ <\chi<\alpha_-$, we continue the estimation
of
$\tilde{V}^+_{t}:=\tilde{V}^+_{t}(X_t;{\theta}_{t+1},\hat{\theta}_{t+2},\ldots,\hat{\theta}_T)$
using the (conditional) H\"older inequality for
$q=\chi/(\lambda \alpha_+ )$ and $1/p+1/q=1$.
\begin{eqnarray*}
\tilde{V}^+_{t}
& \leq &  \tilde{k}_+\left[1+E(\vert X_t+(\theta_{t+1}-
\phi_{t+1})\Delta S_{t+1}\vert^{\lambda \alpha_+ }| \mathcal{F}_t) + \right. \\
 & & \left. E( H^{\lambda \alpha_+ } \vert X_t+(\theta_{t+1}-
\phi_{t+1})\Delta S_{t+1}\vert^{\lambda \alpha_+ }+{H}^{\lambda \alpha_+ }| \mathcal{F}_t)\right]\\
& \leq &
\tilde{k}_+ \left[1+ \vert X_t\vert^{\lambda \alpha_+ } + \vert \theta_{t+1}-\phi_{t+1}
\vert^{\lambda \alpha_+ } E(|\Delta S_{t+1}|^{\lambda \alpha_+ }| \mathcal{F}_t) +
E^{1/p}(H^{\lambda \alpha_+  p}| \mathcal{F}_t)\left( \right.\right.\\
& & \left.\left.
E^{1/q}(\vert X_t\vert^{\chi}| \mathcal{F}_t )+
E^{1/q}(\vert\theta_{t+1}-\phi_{t+1}\vert^{\chi}|\Delta S_{t+1}|^{\chi}| \mathcal{F}_t )\right) +
E({H}^{\lambda \alpha_+ }| \mathcal{F}_t ) \right]\\
& \leq &
\tilde{k}_+ \left[1+ \vert X_t\vert^{\lambda \alpha_+ } + \vert \theta_{t+1}-\phi_{t+1}
\vert^{\lambda \alpha_+ } E(|\Delta S_{t+1}|^{\lambda \alpha_+ }| \mathcal{F}_t) +
E^{1/p}(H^{\lambda \alpha_+  p}| \mathcal{F}_t)\left(
\vert X_t\vert^{\lambda \alpha_+ } \right.\right.\\
& & \left.\left.
+\vert\theta_{t+1}-\phi_{t+1}\vert^{\lambda \alpha_+ }E^{1/q}(|\Delta S_{t+1}|^{\chi}| \mathcal{F}_t )\right)
+
E({H}^{\lambda \alpha_+ }| \mathcal{F}_t ) \right].
\end{eqnarray*}
It follows that, for an appropriate $H_t$ in $\mathcal{W}^+_{t}$,
\begin{eqnarray}
\nonumber \tilde{V}_t(X_t;\theta_{t+1},\ldots,\theta_T)
& \leq & \tilde{k}_-+
H_t\left(1+\vert X_t\vert^{\lambda \alpha_+ }+\vert \theta_{t+1}-\phi_{t+1}\vert^{\lambda \alpha_+ }\right)-\\
\label{figula} & & \tilde{k}_- E\left([ X_t+(\theta_{t+1}-\phi_{t+1})\Delta S_{t+1}+
\sum_{n=t+2}^T (\hat{\theta}_{n}-\phi_n)\Delta S_n-b ]^{\alpha_-}_-| \mathcal{F}_t\right).
\end{eqnarray}

\noindent By Lemma \ref{majdan} below, the event
\[
A:=\{ (\hat{\theta}_n-\phi_n)\Delta S_n\leq 0,\, n\geq t+2;\ (\theta_{t+1}-
\phi_{t+1})\Delta S_{t+1}\leq -\kappa_t \vert \theta_{t+1}-\phi_{t+1}\vert \}
\]
satisfies $P(A\vert\mathcal{F}_t)\geq \tilde{\pi}_t$ with $1/\tilde{\pi}_t\in \mathcal{W}^+_t$, hence considering
\begin{eqnarray}
\label{defF}
F:=\left\{\frac{\vert\theta_{t+1}-\phi_{t+1}\vert\kappa_t}{2}\geq \vert X_t\vert+\vert b\vert \right\}
\end{eqnarray}
we have (recall that $X_{t+1}=X_t+(\theta_{t+1}-\phi_{t+1})\Delta S_{t+1}$),
\begin{eqnarray}\nonumber
1_F E\left([ X_{t+1}+ \sum_{n=t+2}^T (\hat{\theta}_{n}-\phi_n)\Delta S_n-b ]^{\alpha_-}_-| \mathcal{F}_t\right) & \geq &
1_F E\left( 1_A \left( \frac{|\theta_{t+1}-\phi_{t+1}|\kappa_t}{2}\right)^{\alpha_-}| \mathcal{F}_t\right) \\
& \geq &
\left(\frac{|\theta_{t+1}-\phi_{t+1}|\kappa_t}{2}\right)^{\alpha_-}\tilde{\pi}_t 1_F.\label{babar}
\end{eqnarray}
\noindent As a little digression we estimate
\begin{eqnarray}\nonumber
\tilde{V}_t(X_t;\phi_{t+1},\ldots,\phi_T) & = & E\left(\tilde{k}_+
(1+|X_t|^{\lambda \alpha_+ })-\tilde{k}_-[X_t-b]^{\alpha_-}_- | \mathcal{F}_t\right) +\tilde{k}_-\\
\label{barba}
& \geq &
-\tilde{k}_-\vert X_t\vert ^{\alpha_-}-\tilde{k}_-\vert b\vert^{\alpha_-}.
\end{eqnarray}
So on $F$, by \eqref{figula}, \eqref{babar} and \eqref{barba}, using $\vert
X_t\vert^{\lambda \alpha_+ }\leq \vert X_t\vert^{\alpha_-}+1$, we obtain that
\begin{eqnarray*}
\tilde{V}_t(X_t;\theta_{t+1},\ldots,\theta_T) -
\tilde{V}_t(X_t;\phi_{t+1},\ldots,\phi_T)
& \leq & \tilde{k}_-+
H_t\left(2+\vert X_t\vert^{\alpha_-}+\vert \theta_{t+1}-
\phi_{t+1}\vert^{\lambda \alpha_+ }\right)- \\
 & &- \tilde{k}_- \tilde{\pi}_t\left(\frac{|\theta_{t+1}-\phi_{t+1}|\kappa_t}{2}\right)^{\alpha_-}
+\tilde{k}_-\vert X_t\vert ^{\alpha_-}+\tilde{k}_-\vert b\vert^{\alpha_-}\\
 &= & (\tilde{k}_- + H_t)\vert X_t\vert^{\alpha_-} - \frac{\tilde{\pi}_t\tilde{k}_-}{3}\left(\frac{\vert \theta_{t+1}-\phi_{t+1}\vert\kappa_t}{2}\right)^{\alpha_-}
 \\
 & & + 2H_t+\tilde{k}_-\vert b\vert^{ \alpha_-}+ \tilde{k}_- -\frac{\tilde{\pi}_t\tilde{k}_-}{3}\left(\frac{\vert \theta_{t+1}-\phi_{t+1}\vert\kappa_t}{2}\right)^{\alpha_-}
 \\
 & & + H_t\vert \theta_{t+1}-\phi_{t+1}\vert^{\lambda \alpha_+ }-\frac{\tilde{\pi}_t\tilde{k}_-}{3}\left(\frac{\vert \theta_{t+1}-\phi_{t+1}\vert\kappa_t}{2}\right)^{\alpha_-}.
\end{eqnarray*}
\noindent Let us now choose the $\mathcal{F}_t$-measurable random variable $C_t^t$ so large that on the event
\[
\tilde{F}:=\{\vert\theta_{t+1}-\phi_{t+1}\vert > C_t^t[\vert X_t\vert+1] \}
\]
we have
\begin{eqnarray*}
\frac{\vert \theta_{t+1}-\phi_{t+1}\vert\kappa_t}{2} & \geq & \vert X_t\vert+\vert b\vert\quad
\mbox{ (that is,}\ \tilde{F}\subset F\mbox{ holds)}\\
\frac{\tilde{\pi}_t\tilde{k}_-}{3}\left(\frac{\vert \theta_{t+1}-\phi_{t+1}\vert\kappa_t}{2}\right)^{\alpha_-}
 & \geq &
(\tilde{k}_- + H_t)\vert X_t\vert^{\alpha_-},\\
\frac{\tilde{\pi}_t\tilde{k}_-}{3}\left(\frac{\vert \theta_{t+1}-\phi_{t+1}\vert\kappa_t}{2}\right)^{\alpha_-}
& \geq & 2H_t+\tilde{k}_-\vert b\vert^{ \alpha_-}+ \tilde{k}_-\\
\frac{\tilde{\pi}_t\tilde{k}_-}{3}\left(\frac{\vert \theta_{t+1}-\phi_{t+1}\vert\kappa_t}{2}\right)^{\alpha_-}
 & \geq &
H_t\vert \theta_{t+1}-\phi_{t+1}\vert^{\lambda \alpha_+ }.
\end{eqnarray*}
One can easily check that such a $C_t^t$ exists because in order to have the
four preceding inequalities satisfied,
it is sufficient that:
\begin{eqnarray*}
\vert \theta_{t+1}-\phi_{t+1} \vert & \geq & \frac{2}{\kappa_t} \left(
 (\vert X_t\vert+\vert b\vert) +
 \vert X_t \vert \left(\frac{3}{\tilde{\pi}_t\tilde{k}_-}(\tilde{k}_- + H_t)\right)^{1/\alpha_-} +
\left(\frac{3}{\tilde{\pi}_t\tilde{k}_-}
(2H_t+\tilde{k}_-\vert b\vert^{ \alpha_-}+ \tilde{k}_-\right)^{1/\alpha_-}\right)\\
& &
+ \left(\frac{3 \times 2^{\alpha_-}H_t}{\tilde{\pi}_t\tilde{k}_-\kappa_t^{\alpha_-}}\right)^{\frac1{\alpha_--\lambda \alpha_+ }},
%\\ & > & C_t^t[\vert X_t\vert+1],
\end{eqnarray*}
hence one can clearly find $C_t^t\in\mathcal{W}_t^+$ such that
$C_t^t [\vert X_t\vert+1]$ is greater than the right-hand side of
the above inequality.
%where $C_t^t$  is in
%$\mathcal{W}_t^+$.
So on $\tilde{F}$ we have,
\begin{eqnarray}
\label{ineqvtildethetaphi}
\tilde{V}_t(X_t;\theta_{t+1},\ldots,\theta_T) -
\tilde{V}_t(X_t;\phi_{t+1},\ldots,\phi_T) \leq 0
\end{eqnarray}
%On $\tilde{F}$ we have, using $\vert
%X_t\vert^{\lambda \alpha_+ }\leq \vert X_t\vert^{\alpha_-}+1$ and
%thanks to \eqref{figula}, \eqref{babar} and \eqref{barba}:
%\begin{eqnarray}
%\nonumber
%\tilde{V}_t(X_t;\theta_{t+1},\ldots,\theta_T) -
%\tilde{V}_t(X_t;\phi_{t+1},\ldots,\phi_T)
%& \leq & \tilde{k}_-+
%H_t\left(2+\vert X_t\vert^{\alpha_-}+\vert \theta_{t+1}-
%\phi_{t+1}\vert^{\lambda \alpha_+ }\right)- \\
%\nonumber
%& & \tilde{k}_-
%\left(\frac{|\theta_{t+1}-\phi_{t+1}|\kappa_t}{2}\right)^{\alpha_-}\tilde{\pi}_t
%+ \tilde{k}_-\vert X_t\vert ^{\alpha_-}+\tilde{k}_-\vert b\vert^{\alpha_-} \\
%\nonumber
%& \leq & (\tilde{k}_- + H_t)\vert X_t\vert^{\alpha_-} + H_t\vert \theta_{t+1}-\phi_{t+1}\vert^{\lambda \alpha_+ }
%+ 2H_t \\
%\nonumber
%& & + \tilde{k}_-\vert b\vert^{\alpha_-}-\tilde{k}_-
%\left(\frac{|\theta_{t+1}-\phi_{t+1}|\kappa_t}{2}\right)^{\alpha_-}\tilde{\pi}_t+\tilde{k}_- \\
%\label{ineqvtildethetaphi}
%& \leq & 0.
%\end{eqnarray}
Consequently, defining
\begin{eqnarray*}
\tilde{\theta}_{t+1} & := & \phi_{t+1}1_{\tilde{F}}+\theta_{t+1}1_{\tilde{F}^c},\\
\tilde{\theta}_{n} & :=& \phi_{n}1_{\tilde{F}}+\hat{\theta}_{n}1_{\tilde{F}^c},\quad n=t+2, \ldots,  T,
\end{eqnarray*}
we have, using \eqref{ineqvtilde} and \eqref{ineqvtildethetaphi},
\[
\tilde{V}_t(X_t;\theta_{t+1},\ldots,\theta_T)\leq \tilde{V}_t(X_t;\tilde{\theta}_{t+1},
\ldots,\tilde{\theta}_T)\mbox{ a.s.}.
\]
By construction,
\[
\vert\tilde{\theta}_{t+1}-\phi_{t+1}\vert\leq C_t^t[\vert X_t\vert+1],
\]
and, for $n\geq t+2$,
\begin{eqnarray*}
\vert\tilde{\theta}_{n}-\phi_n\vert & = &
1_{\tilde{F}^c} \vert\hat{\theta}_{n}-\phi_n\vert
%\leq  \vert\hat{\theta}_{n}-\phi_n\vert
\leq {1_{\tilde{F}^c}  C^{t+1}_{n-1} [\vert X_t+({\theta}_{t+1}-\phi_{t+1})
\Delta S_{t+1}\vert+1]}\\
& \leq & {1_{\tilde{F}^c}  C^{t+1}_{n-1} [\vert X_t+(\tilde{\theta}_{t+1}-\phi_{t+1})
\Delta S_{t+1}\vert+1]} \\
& \leq &
C^{t+1}_{n-1} [\vert X_t\vert + C_t^t(\vert X_t\vert+1)\vert \Delta S_{t+1}\vert +1]=C^t_{n-1}(\vert X_t\vert+1),
\end{eqnarray*}
where $
C^t_{n-1}:=
C^{t+1}_{n-1} (C_t^t\vert \Delta S_{t+1}\vert +1)
$
for $n\geq t+2$. Clearly, $C_{n-1}^t\in\mathcal{W}^+_{n-1}$. To
conclude the proof it remains to check that
$(\tilde{\theta}_{t+1},\ldots,
\tilde{\theta}_{T})\in\tilde{\mathcal{A}}_t(X_t)$. As by
hypothesis
$({\theta}_{t+1},\ldots,{\theta}_T)\in\tilde{\mathcal{A}}_{t}(X_t)$, we get
from \eqref{ineqvtilde} that $\tilde{V}^-_{t}(X_{t};\theta_{t+1},
\hat{\theta}_{t+2},\ldots,\hat{\theta}_T)< \infty$.
%As on $\{|X_t|=\infty \}$, $\tilde{F}=\emptyset$, we get
%that
%So that, as $E\vert X_t\vert^{\alpha_-}<\infty$,
Finally,
$$
\tilde{V}_t^-(X_t;\tilde{\theta}_{t+1},\ldots,\tilde{\theta}_T) =
1_{\tilde{F}} \tilde{k}_- \left((X_t-b)^{\alpha_-}_- -\tilde{k}_-\right) +
1_{\tilde{F}^c} \tilde{V}^-_{t}(X_{t};\theta_{t+1},
\hat{\theta}_{t+2},\ldots,\hat{\theta}_T)
<\infty\quad\mbox{a.s.}
$$

In the course of this proof we relied on Lemma \ref{majdan} below.

\begin{lemma}\label{majdan} Assume that Assumption \ref{marche}  holds true.
Then there exists $\tilde{\pi}_t>0$ with $1/\tilde{\pi}_t\in\mathcal{W}^+_t$
such that
\[
P(
(\theta_{t+1}-
\phi_{t+1})\Delta S_{t+1}\leq -\kappa_t \vert \theta_{t+1}-\phi_{t+1}\vert,\
(\hat{\theta}_n-\phi_n)\Delta S_n\leq 0,\, n=t+2, \ldots, T| \mathcal{F}_t)\geq\tilde{\pi}_t.
\]
\end{lemma}
%\begin{proof}
%See Appendix \ref{appendix majdan}.
%\end{proof}
%\subsubsection{Proof of Lemma \ref{majdan}}
%\label{appendix majdan}
\begin{proof}
Define the events
\begin{eqnarray*}
A_{t+1} & :=& \{ (\theta_{t+1}-
\phi_{t+1})\Delta S_{t+1}\leq -\kappa_t \vert \theta_{t+1}-\phi_{t+1}\vert\},\\
A_n& := & \{(\hat{\theta}_n-\phi_n)\Delta S_n\leq 0\},\quad t+2\leq n\leq T.
\end{eqnarray*}
We prove, by induction, that for $m\geq t+1$,
\begin{eqnarray}\label{hypo}
E(1_{A_{t+1}}\ldots 1_{A_m}| \mathcal{F}_t)\geq \tilde{\pi}_t(m)
\end{eqnarray}
for some $\tilde{\pi}_t(m)$ with $1/\tilde{\pi}_t(m)\in\mathcal{W}_t^+$.
For $m=t+1$ this is just \eqref{rrr}.
Let us assume that \eqref{hypo} has been shown for $m-1$, we will establish it for $m$.
\begin{eqnarray*}
E(1_{A_m}\ldots 1_{A_{t+1}}| \mathcal{F}_t) & = & E(E(1_{A_m}| \mathcal{F}_{m-1})
1_{A_{m-1}}\ldots 1_{A_{t+1}}| \mathcal{F}_t)\\
& \geq & E(\pi_{m-1} 1_{A_{m-1}}\ldots 1_{A_{t+1}}| \mathcal{F}_t) \\
& \geq &
\frac{E^2(1_{A_{m-1}}\ldots 1_{A_{t+1}}| \mathcal{F}_t)}{E(1/\pi_{m-1}| \mathcal{F}_t)}
\geq
\frac{\tilde{\pi}^2_t(m-1)}{E(1/\pi_{m-1}| \mathcal{F}_t)}:=\tilde{\pi}_t(m-1)
\end{eqnarray*}
by the (conditional) Cauchy inequality.
Here $1/\tilde{\pi}_t(m-1)\in\mathcal{W}^+_t$ by the induction hypothesis,
$E(1/\pi_{m-1}\vert \mathcal{F}_t)\in\mathcal{W}^+_t$ (since $1/\pi_{m-1}\in\mathcal{W}^+$)
and the statement follows.
\end{proof}

\subsubsection{Proof of Lemma \ref{oslo}}
\label{appendix oslo}
Fix $c\in\mathbb{R}$ and $\chi,\iota ,o$ satisfying $\lambda\alpha_+<\chi<\iota
<o<\alpha_-$. Let $X_t \in \Xi_t^1$ with $E\vert
X_t\vert^{o}<\infty$ and  $(\theta_{t+1}, \ldots,
\theta_{T})\in\tilde{\mathcal{A}}_t(X_t)$
such that
\[
E\tilde{V}_t(X_t;\theta_{t+1},\ldots,\theta_T)\geq c.
\]
Let $X_{t+1}:=X_t+(\theta_{t+1}-\phi_{t+1})\Delta S_{t+1}$. By Lemma \ref{crux}, there exists
$C_{n}^{t+1} \in {\cal W}_n^+$, $t+1 \leq n \leq T-1$, and
$(\hat{\theta}_{t+2},\ldots,\hat{\theta}_{T})
\in\tilde{\mathcal{A}}_{t+1}(X_{t+1})$ such that
$$
\vert\hat{\theta}_n-\phi_n\vert\leq C^{t+1}_{n-1} [\vert X_{t+1}\vert+1],
$$
for $n=t+2,\ldots, T$ and $$
\tilde{V}_{t+1}(X_{t+1};\theta_{t+2},\ldots,\theta_T)\leq \tilde{V}_{t+1}(X_{t+1};\hat{\theta}_{t+2},
\ldots,\hat{\theta}_T).
$$
We can obtain equations \eqref{ineqvtilde} and \eqref{figula} just like in the proof of Lemma  \ref{crux}.
Furthermore, using \eqref{babar}, we get (recall \eqref{defF} for the definition of $F$) :
\begin{eqnarray}
\nonumber
E\tilde{V}_t(X_t;\theta_{t+1},\ldots,\theta_T) & \leq & E(H_t(1+\vert X_t\vert^{\lambda \alpha_+ }+
\vert \theta_{t+1}-\phi_{t+1}\vert^{\lambda \alpha_+ }))\\ & &
\label{cars}
-\tilde{k}_-
E\left(1_F \left(\frac{\vert \theta_{t+1}-\phi_{t+1}\vert\kappa_t}{2}\right)^{\alpha_-}\tilde{\pi}_t\right)
+\tilde{k}_-.
\end{eqnarray}
We now push further estimations in this last equation.

We may estimate, using the H\"older inequality for $p=\alpha_-/o$
and its conjugate $q$,
\begin{eqnarray*}
E\left(1_F \left(\frac{\vert \theta_{t+1}-\phi_{t+1}\vert\kappa_t}{2}\right)^{\alpha_-}\tilde{\pi}_t\right)
\geq
\frac{E^{p}
\left(1_F \left(\frac{\vert \theta_{t+1}-\phi_{t+1}\vert\kappa_t}{2}\right)^{o}\tilde{\pi}^{1/p}_t
\frac1{\tilde{\pi}^{1/p}_t}\right)}
{E^{p/q}\left(\frac{1}{\tilde{\pi}^{q/p}_t}\right)}.
\end{eqnarray*}
The denominator here will be denoted $C$ in the sequel. By Lemma \ref{majdan},
$C<\infty$.

Now let us note the trivial fact that for random variables $X,Y\geq
0$ such that $EY^{o}\geq 2 EX^{o}$ one has $E[1_{\{Y\geq X\}}
Y^{o}]\geq \frac12 EY^{o}$.

It follows that if
\begin{equation}\label{feltetel1}
E\left(\frac{\vert \theta_{t+1}-\phi_{t+1}\vert\kappa_t}{2}\right)^{o}\geq
2 E(\vert X_t\vert+\vert b\vert)^{o}
\end{equation}
holds true then, applying the trivial $x\leq x^p+1,\, x\geq 0$,
\begin{eqnarray*}
\frac{E^{p}
\left(1_F \left(\frac{\vert \theta_{t+1}-\phi_{t+1}\vert\kappa_t}{2}\right)^{o}\right)}
{C} & \geq &
\frac{E^{p}
\left(\left(\frac{\vert \theta_{t+1}-\phi_{t+1}\vert\kappa_t}{2}\right)^{o}\right)}
{2^pC}\\
& \geq & \frac{E\left(\frac{\vert \theta_{t+1}-\phi_{t+1}\vert\kappa_t}{2}\right)^{o}-1}
{2^pC}=
c_1 E({\vert \theta_{t+1}-\phi_{t+1}\vert\kappa_t})^{o}-c_2
\end{eqnarray*}
with suitable $c_1,c_2>0$. Using again H\"older's inequality with
$p=o/\iota $ and its conjugate $q$,
\begin{eqnarray}
\label{togo}
E({\vert \theta_{t+1}-\phi_{t+1}\vert\kappa_t})^{o} & \geq &
\frac{E^{p}{\vert \theta_{t+1}-\phi_{t+1}\vert}^{\iota }}
{E^{p/q}\left(\frac{1}{\kappa_t^{\iota q }}\right)}\geq
\frac{E{\vert \theta_{t+1}-\phi_{t+1}\vert}^{\iota }-1}
{E^{p/q}\left(\frac{1}{\kappa_t^{\iota q}}\right)}.
\end{eqnarray}
%\begin{eqnarray*}
%c_1 E[({\vert \theta_{t+1}-\phi_{t+1}\vert\kappa_t})^{o}]-c_2 & \geq &
%c_1 \frac{E^{p}{\vert \theta_{t+1}-\phi_{t+1}\vert}^{\iota }}
%{E^{p/q}\left[\frac{1}{\kappa_t^{\iota q }}\right]}-c_2 \geq
%c_1 \frac{E{\vert \theta_{t+1}-\phi_{t+1}\vert}^{\iota }-1}
%{E^{p/q}\left[\frac{1}{\kappa_t^{\iota q}}\right]}-c_2.
%\end{eqnarray*}
With suitable $c_1',c_2'>0$, we get, whenever \eqref{feltetel1} holds, that
\begin{eqnarray}\label{togo2}
E\left(1_F \left(\frac{\vert \theta_{t+1}-\phi_{t+1}\vert\kappa_t}{2}\right)^{\alpha_-}\tilde{\pi}_t\right)
& \geq &
 c_1' E{\vert \theta_{t+1}-\phi_{t+1}\vert}^{\iota }-c_2'.
\end{eqnarray}

\noindent Estimate also, with $p:=\chi/(\lambda \alpha_+ )$,
\begin{eqnarray}\nonumber
E\left(H_t(1+\vert X_t\vert^{\lambda\alpha_+}+\vert\theta_{t+1}-\phi_{t+1}\vert^{\lambda\alpha_+})\right)
& \leq &
E^{1/q}[H_t^q] [1+E^{1/p}\vert X_t\vert^{\chi}+
E^{1/p}\vert\theta_{t+1}-\phi_{t+1}\vert^{\chi}]\\
\nonumber & \leq &  E^{1/q}[H_t^q] [3+E\vert X_t\vert^{\chi}+
E\vert\theta_{t+1}-\phi_{t+1}\vert^{\chi}]\\
\label{fontoslesz} & \leq & \tilde{c} [1+E\vert X_t\vert^{o}+
E\vert\theta_{t+1}-\phi_{t+1}\vert^{\chi}],
\end{eqnarray}
with some $\tilde{c}>0$, using that $x^{\chi}\leq x^{o}+1$,
$x^{1/p}\leq x+1$, for $x\geq 0$. Furthermore, H\"older's inequality
with $p=\iota /\chi$ gives
\begin{eqnarray*}
E\vert\theta_{t+1}-\phi_{t+1}\vert^{\chi}\leq
E^{\chi/\iota }\vert\theta_{t+1}-\phi_{t+1}\vert^{\iota }.
\end{eqnarray*}

\noindent It follows that whenever
\begin{equation}\label{qqq}
\left(E{\vert \theta_{t+1}-\phi_{t+1}\vert}^{\iota }\right)^{1-\chi/\iota }\geq
\frac{2\tilde{c}}{c_1' \tilde{k}_-},
\end{equation}
one also has
\begin{equation}\label{benin}
\tilde{c}E\vert\theta_{t+1}-\phi_{t+1}\vert^{\chi}\leq \frac{c_1'\tilde{k}_-}{2}E{\vert \theta_{t+1}-\phi_{t+1}\vert}^{\iota }.
\end{equation}

\noindent Finally consider the condition
\begin{equation}\label{ppp}\frac{c_1'\tilde{k}_-}{2}
E{\vert \theta_{t+1}-\phi_{t+1}\vert}^{\iota }\geq \tilde{c}[1+E\vert X_t\vert^{o}]
%+c_2'\tilde{k}_- 1_{c>0}
+(c_2'\tilde{k}_--c+1)+\tilde{k}_-.
%1_{c\leq0}.
\end{equation}
It is easy to see that we can find some $K_t$, large enough, such
that $E\vert\theta_{t+1}-\phi_{t+1}\vert^{\iota }\geq K_t[E\vert
X_t\vert^{o}+1]$ implies that \eqref{feltetel1} (recall
\eqref{togo}), \eqref{qqq}, \eqref{ppp} all hold true.
%It is clear from \eqref{togo} that if $K_t$ is large enough and
%$E\vert\theta_{t+1}-\phi_{t+1}\vert^{\iota }\geq K_t[E\vert X_t\vert^{o}+1]$
%then \eqref{feltetel1}, \eqref{qqq}, \eqref{ppp}
%all hold true.
So in this case we have, from \eqref{cars}, \eqref{fontoslesz}, \eqref{benin}, \eqref{togo2} and
\eqref{ppp},
\begin{eqnarray*}
E\tilde{V}_t(X_t;\theta_{t+1},\ldots,\theta_T) & \leq & \tilde{c}[1+E\vert X_t\vert^{o}]
+\frac{c_1'\tilde{k}_-}{2}
E{\vert \theta_{t+1}-\phi_{t+1}\vert}^{\iota } \\
& & -c_1'\tilde{k}_-
E{\vert \theta_{t+1}-\phi_{t+1}\vert}^{\iota }+c_2'\tilde{k}_-+\tilde{k}_- \\
& \leq &
%-c_2'\tilde{k}_- 1_{c>0}
-(c_2'\tilde{k}_--c+1)
%1_{c\leq0}
+ c_2'\tilde{k}_-
< c.
\end{eqnarray*}
From this the statement of Lemma \ref{oslo} follows.

\subsubsection{Proof of Proposition \ref{kh}}
\label{appendixkh} Assumptions \ref{marche}, \ref{as1}, \ref{as2}
are clearly met (with $\phi\equiv 0$). Theorem \ref{ajtothm} implies
that $M_n<\infty$ for all $n$.
{Let $c=V(0;0)$, from Lemma \ref{eel} we get that
$\tilde{V}(0;0)=E\tilde{V}_0(0;0) \geq c$. Now fix some
constant $\iota>0$ such that $\lambda \alpha_- < \iota < \alpha_+$.
Looking at the end of the proof of
Lemma \ref{oslo} and remarking that
$\tilde{\mathcal{A}}_0(0)=\mathcal{A}_n$, we get that
%for $c=V(0;0)$
there exists constant
%s $\iota>0$ and
$K \geq 0$ such that if
$\theta\in\mathcal{A}_n$ with $E\vert\theta\vert^{\iota} >K$ then
$E\tilde{V}_0(0;\theta)<c$}. From Lemma \ref{eel} again,
$$V(0;\theta)\leq \tilde{V}(0;\theta)=E\tilde{V}_0(0;\theta)<c=V(0;0)$$
and hence $\theta$ is
suboptimal.

It follows from the above argument that the optimization can be constrained to the smaller
domains $D_n:=\{\theta\in\mathcal{A}_n: E\vert\theta\vert^{\iota}\leq K\}$
for each $n$. As the probability space is finite, the space
of $\mathcal{H}_n$-measurable random variables (equipped with the topology
of convergence in probability) can be identified with a finite-dimensional
Euclidean space where $D_n$ is a compact set. Since the objective function $V(0;\cdot)$ is easily
seen to be continuous the supremum $M_n$ is attained by some
strategy $\theta^*_n$, $n\geq 0$.

Let $\Lambda$ be the (finite) range of the random variable
$|\theta_n^*|$. By Lemma \ref{ret}, $\Lambda$ contains a nonzero
element. Let $a$ denote the smallest such element and $b$ the
largest one, we get that {either $\Lambda=\{0,a_0,\ldots,a_n\}$ or
$\Lambda=\{a_0,\ldots,a_n\}$}, with $a=a_0<a_1<\ldots<a_n=b$. Let us
introduce the notations $A_+:=\{\theta_n^*=a\}$,
$A_-:=\{\theta_n^*=-a\}$, $A:=A_+\cup A_-=\{|\theta_n^*|=a\}$.
%Note that
%\begin{eqnarray}\label{torch}
%\{\theta_n^*\Delta S_1=a\}=(A_+\cap \{\Delta S_1=1\})\cup (A_-\cap \{\Delta S_1=-1%\}).
%\end{eqnarray}
%and denote $p_n:=P(\theta_n^*\Delta S_1=a)$. By the choice of $a$, $p_n>0$.

%Let and hence also a strictly positive element (note that $\theta_n^*$ and $\Delta S_1$ are
%independent and the latter takes both positive and negative values). Let the smallest such element be $a>0$.
% and denote the nearest element
%in the range that is strictly greater than $a$ by $b$. (Assume for the moment
%that there is such a $b$.) Let $p:=P(\theta_n^*\Delta S_1>0)>0$ and
%$0<

For each $\delta\geq 0$ we will define a $\mathcal{H}_{n+1}$-measurable
strategy $\Theta_{n+1}(\delta)$ which has a strictly better performance
than $\theta_n^*$ for a suitable choice of $\delta$, $i.e.$
$M_n=V(0;\theta_n^*)<V(0;\Theta_{n+1}(\delta))\leq M_{n+1}$.

Let
$\Theta_{n+1}(\delta)=a+\delta$ on $A_+\cap \{\epsilon_{n+1}=1\}$,
$\Theta_{n+1}(\delta)=-a-\delta$ on $A_-\cap \{\epsilon_{n+1}=1\}$,
$\Theta_{n+1}(\delta)=
a-\delta$ on $A_+\cap \{\epsilon_{n+1}=-1\}$,
$\Theta_{n+1}(\delta)=
-a+\delta$ on $A_-\cap \{\epsilon_{n+1}=-1\}$, $\Theta_{n+1}(\delta)=\theta_n^*$
outside $A$. In particular, $\theta_n^*=\Theta_{n+1}(0)$. This definition implies that $|\Theta_{n+1}(\delta)|=a + \delta \epsilon_{n+1}$ on $A$ and $|\Theta_{n+1}(\delta)|=|\theta_n^*|$ outside $A$.
So from \eqref{saphir} and using independence of  $\epsilon_{n+1}$ and $\theta_n^*$, one gets
\begin{eqnarray}\nonumber
V^-(0;\Theta_{n+1}(\delta)) & = &  \frac12 \left(E(1_A(|\theta_n^*|+ \delta \epsilon_{n+1}))+E1_{A^c}|\theta_n^*| \right)\\
\nonumber & = & \frac12 \left(E(1_A|\theta_n^*|)+ \delta P(A)E\epsilon_{n+1} +E1_{A^c}|\theta_n^*| \right) = \frac12E |\theta_n^*|\\
\label{fd} & = & V^-(0;\theta_n^*).
\end{eqnarray}

%We now define, for each $n$, a $\mathcal{H}_{n+1}$-measurable functions $\Theta_n(\delta)$, $\delta\geq 0$
%as follows: $\Theta_n(\delta)\Delta S_1:=\theta^*_n \Delta S_1$ when
%$\theta^*_n \Delta S_1\neq a$; $\Theta_n(\delta) \Delta S_1=a-\delta$ on
%$\{\theta^*_n \Delta S_1\neq a\}\cap\{\varepsilon_{n+1}=1\}$ and
%$\Theta_n (\delta) \Delta S_1=a+\delta$ on
%$\{\theta^*_n \Delta S_1\neq a\}\cap\{\varepsilon_{n+1}=-1\}$. Notice that
%$\Theta_n(0)=\theta_n^*$.

Now we are looking at $V^+(0;\Theta_{n+1}(\delta))$.
First let us consider the case where $a=b$, then $A=\{|\theta_n^*|=a\}$ and $A^c=\{|\theta_n^*|=0\}$. Take $0\leq \delta<a$. Note that in this case $|\Theta_{n+1}(\delta)|$
may take only the values $0,a-\delta,a+\delta$. So it follows that from \eqref{rubis} and using independence of  $\epsilon_{n+1}$ and $\theta_n^*$,
\begin{eqnarray*}
V^+(0;\Theta_{n+1}(\delta)) & = &
\sqrt{\frac12}\int_0^{\infty}\sqrt{E1_A 1_{(a + \delta \epsilon_{n+1})^{1/4}\geq y}+E1_{A^c}1_{0 \geq y}}dy \\
& = &
\sqrt{\frac12}\int_0^{\infty}\sqrt{P(A) \left(\frac12 1_{(a + \delta)^{1/4}\geq y}+\frac12 1_{(a -\delta)^{1/4}\geq y}\right)}dy \\
& = &
\sqrt{\frac12}\left(\int_0^{(a -\delta)^{1/4}}\sqrt{P(A)} dy + \int_{(a -\delta)^{1/4}}^{(a +\delta)^{1/4}}\sqrt{\frac12P(A)}\right)\\
& = &
\sqrt{\frac12}\left((a -\delta)^{1/4}\sqrt{P(A)} + \left((a +\delta)^{1/4}-(a -\delta)^{1/4}\right)\sqrt{\frac12P(A)}\right).
\end{eqnarray*}
We have that $P(A)=P(|\theta_n^*|=a)>0$ by the choice of $a$ and
{ $P(A)$ does not depend on $\delta$}. So one can directly check that
$V^+(0;\Theta_{n+1}(\delta))$ is continuously differentiable in
$\delta$ (in a neighborhood of $0$) and
\[
\frac{\partial}{\partial \delta}V^+(0;\Theta_{n+1}(\delta))\vert_{\delta=0}=(\sqrt2 -1)\frac{\sqrt2}8a^{-3/4}\sqrt{P(A)}>0.
\]
Hence, for $\delta>0$ small enough,
\begin{equation}\label{rag}
V^+(0;\Theta_{n+1}(\delta))>V^+(0;\theta_n^*)=V^+(0;\Theta_{n+1}(0)).
\end{equation}

Now let us turn to the case where $a<b$. Then $A=\{|\theta_n^*|=a\}$
and $A^c=\{|\theta_n^*|\in\{0,a_1,\ldots,a_n\}\}$. We may write (for
$\delta$ small enough such that $a-\delta>0$ and $a+\delta<a_1$),
\begin{eqnarray*}
V^+(0;\Theta_{n+1}(\delta)) & = & \sqrt{\frac12}\int_0^{\infty}\sqrt{E1_A 1_{(a + \delta \epsilon_{n+1})^{1/4}\geq y}+E1_{A^c}1_{|\theta_n^*|^{1/4} \geq y}}dy \\
& = &
\sqrt{\frac12}\int_0^{\infty}\sqrt{P(A) \left(\frac12 1_{(a + \delta)^{1/4}\geq y}+\frac12 1_{(a -\delta)^{1/4}\geq y}\right)
+\sum_{i=1}^nE1_{|\theta_n^*|=a_i}1_{a_i^{1/4} \geq y}}dy \\
& = &
\sqrt{\frac12}\left(
\int_0^{(a -\delta)^{1/4}}\sqrt{P(A)+P(|\theta_n^*|\geq a_{1})} dy + \int_{(a -\delta)^{1/4}}^{(a +\delta)^{1/4}}\sqrt{\frac12P(A)+P(|\theta_n^*|\geq a_{1})} dy \right.+ \\
& & \left.
\int_{(a +\delta)^{1/4}}^{a_1^{1/4}} \sqrt{P(|\theta_n^*|\geq a_{1})} dy + \sum_{i=1}^{n-1} \int_{a_i^{1/4}}^{a_{i+1}^{1/4}}\sqrt{P(|\theta_n^*|\geq a_{i+1})}\right) dy\\
& = &
\sqrt{\frac12}\left(
(a -\delta)^{1/4} \sqrt{P(A)+P(|\theta_n^*|\geq a_{1})}+ \left((a +\delta)^{1/4}-(a -\delta)^{1/4}\right)\sqrt{\frac12P(A)+P(|\theta_n^*|\geq a_{1})}  \right.+ \\
& & \left.
\left(a_1^{1/4}-(a +\delta)^{1/4}\right) \sqrt{P(|\theta_n^*|\geq a_{1})}  + \sum_{i=1}^{n-1}\left(a_{i+1}^{1/4}-a_i^{1/4}\right)\sqrt{P(|\theta_n^*|\geq a_{i+1})}\right).
\end{eqnarray*}
%Note that $P(A^c)=P(|\theta_n^*|\geq a_{1})>P(|\theta_n^*|\geq a_{2})>\ldots>P(|\theta_n^*|\geq a_{n})=P(|\theta_n^*|=b)>0$
%and all those quantities do not depend on $\delta$.
Note that $P(A)$, $P(|\theta_n^*|\geq a_{1})$, $P(|\theta_n^*|\geq a_{2})$,\ldots, $P(|\theta_n^*|\geq a_{n})=P(|\theta_n^*|=b)$ do not depend on $\delta$ and that $P(A)>0$ by the choice of $a$.
% and $b$ (recall that $P(A)+P(|\theta_n^*|\geq a_{1})+P(|\theta_n^*|=0)=1$).
Again, one can directly check that $V^+(0;\Theta_{n+1}(\delta))$ is continuously
differentiable in $\delta$ (in a neighborhood of $0$) and
\[
\frac{\partial}{\partial \delta}V^+(0;\Theta_{n+1}(\delta))\vert_{\delta=0}
=\frac{\sqrt2}8 a^{-3/4} \left(-\sqrt{P(A)+P(|\theta_n^*|\geq a_{1})}+\sqrt2\sqrt{P(A)+2P(|\theta_n^*|\geq a_{1})}-\sqrt{P(|\theta_n^*|\geq a_{1})}\right).
\]

By direct computation, as $P(A)>0$, one get that $\frac{\partial}{\partial \delta}V^+(0;\Theta_{n+1}(\delta))\vert_{\delta=0}>0$ and for $\delta$ small enough, \eqref{rag} holds true. Fix such a $\delta$, recall from \eqref{fd} that  $V^-(0;\Theta_{n+1}(\delta))=V^-(0;\theta_{n+1}^*)$ so $M_n=V(0;\theta_n^*)<V(0;\Theta_{n+1}(\delta))\leq M_{n+1}$ and Proposition \ref{kh} is proved.

\subsubsection{Proof of Lemma \ref{porto}}
\label{appendix porto}

Take $\tau:=\alpha_{T}<\alpha_{T-1}<\ldots<\alpha_1<\alpha_0:=\alpha_-$.
We first prove, by induction on $t$, that $X_t:=X_0+\sum_{j=1}^t
(\theta_j-\phi_j)\Delta S_{j}$, $t\geq 0$ satisfy
\[
E\vert X_t\vert^{\alpha_t}\leq C_t [E\vert X_0\vert^{\alpha_-}+1],
\]
for suitable $C_t>0$. For $t=0$ this is trivial. Assuming it for $t$
we will show it for $t+1$. We first remark that
\[
E\tilde{V}_t(X_t;\theta_{t+1},\ldots,\theta_T)= E\tilde{V}_0(X_0;\theta_{1},\ldots,\theta_T)\geq c
\]
and that by the induction hypothesis $E\vert
X_t\vert^{\alpha_t}<\infty$ holds. As
$\theta\in\tilde{\mathcal{A}}(X_0)\subset\tilde{\mathcal{A}}_0(X_0)$,
$(\theta_{t+1}, \ldots, \theta_{T})\in\tilde{\mathcal{A}}_t(X_t)$
(see Lemma \ref{tavasz}). Thus Lemma \ref{oslo} applies with the
choice $\iota :=(\alpha_{t+1}+\alpha_t)/2$ and $o:=\alpha_t$, and we
can estimate, using H\"older's inequality with $p:=\iota
/\alpha_{t+1}$ (and its conjugate number $q$), plugging in the
induction hypothesis:
\begin{eqnarray*}
E\vert X_{t+1}\vert^{\alpha_{t+1}} & = &
E\vert X_t+(\theta_{t+1}-\phi_{t+1})\Delta S_{t+1}\vert^{\alpha_{t+1}}\\
 & \leq & E\vert X_t\vert^{\alpha_{t+1}}+E\vert (\theta_{t+1}-\phi_{t+1})\Delta S_{t+1}\vert^{\alpha_{t+1}}
\\
 & \leq & E\vert X_t\vert^{\alpha_t}+1+
E^{1/p}\vert (\theta_{t+1}-\phi_{t+1})\vert^{\iota } E^{1/q}\vert\Delta S_{t+1}\vert^{q\alpha_{t+1}}\\
& \leq & E\vert X_t\vert^{\alpha_t}+1+ C \left(E\vert (\theta_{t+1}-\phi_{t+1})\vert^{\iota } +1 \right)\\
& \leq &  E\vert X_t\vert^{\alpha_t}+1+
C \left(K_t(E\vert X_t\vert^{\alpha_t}+1) +1 \right)\\
& \leq & (1+CK_t)C_t\left(E\vert X_0\vert^{\alpha_-}+1\right)+1+C+CK_t
\end{eqnarray*}
with $C:=E^{1/q}\vert\Delta S_{t+1}\vert^{q\alpha_{t+1}}$, this
proves the induction hypothesis for $t+1$.

Now let us observe that, by Lemma \ref{oslo} (with $\iota
=\alpha_{t+1},o=\alpha_t$),
\begin{eqnarray*}
E\vert \theta_{t+1}-\phi_{t+1}\vert^{\tau}
& \leq  &
E\vert \theta_{t+1}-\phi_{t+1}\vert^{\alpha_{t+1}} +1 \\
& \leq  & K_t[E\vert X_t\vert^{\alpha_t}+1]+1\leq K_t [C_t(E\vert X_0\vert^{\alpha_-}+1)+1]+1,
\end{eqnarray*}
concluding the proof.

\subsection{Auxiliary results}\label{het}
We start with simple observations.
\begin{lemma}\label{ratatouille}
Let $(X_n)_n$ be a tight sequence of random variables in $\mathbb{R}^N$. Then, for any random variable $X$ in $\mathbb{R}^N$\\
(i) $(X_n+X)_n$ is  a tight sequence of random variables in $\mathbb{R}^N$.\\
(ii) $(X_n,X)_n$ is  a tight sequence of random variables in
$\mathbb{R}^{2N}$.
\end{lemma}
\begin{proof}
Fix some $\eta>0$, there exists some $k_0>0$ such that $P(|X_n|>k_0)
< \eta/2$, for each $n$. As $\cap_m \{|X| >m\}= \emptyset$, there
exists $k_1$ such that $P(|X|>k_1) < \eta/2$. Thus, we obtain that
$$P(|X_n+X| \leq k_0+k_1) \geq P(|X_n| \leq k_0) +P(|X| \leq k_1) -1 > 1-\eta,$$
showing (i). It is clear that $(X_n,0)_n$ is a tight sequence of
random variables in $\mathbb{R}^{2N}$, so from (i), we deduce (ii).
\end{proof}

\begin{lemma}\label{magi}
If $Y\in\mathcal{W}^+$ then
\[
\int_0^{\infty} P^{\delta}(Y\geq y)dy<\infty,
\]
for all $\delta>0$.
\end{lemma}
\begin{proof}
As by Chebishev's inequality and $Y\in\mathcal{W}^+$,
\[
P(Y\geq y)\leq M(N)y^{-N},\quad y>0,
\]
for all $N>0$, for a constant $M(N):=EY^N$, we can choose $N$ so
large to have $N\delta>1$, showing that the integral in question is
finite.
\end{proof}

The following Lemmata should be fairly standard. We nonetheless
included their proofs since we could not find an appropriate
reference.

\begin{lemma}\label{ode}
Let $E$ be uniformly distributed on $[0,1]$. Then for each $l\geq 1$
there are measurable $f_1,\ldots, f_l:[0,1]\to [0,1]$ such that
$f_1(E),\ldots,f_l(E)$ are independent and uniform on $[0,1]$.
\end{lemma}
\begin{proof} We first recall that if $\mathcal{Y}_1,\mathcal{Y}_2$ are uncountable
Polish spaces then they are Borel isomorphic, i.e. there is a
bijection $\psi:\mathcal{Y}_1\to\mathcal{Y}_2$ such that
$\psi,\psi^{-1}$ are measurable (with respect to the respective
Borel fields); see e.g. page 159 of \cite{dm}.

Fix a Borel-isomorphism $\psi:\mathbb{R}\to [0,1]^l$ and define the
probability $\kappa(A):=\lambda_l(\psi(A))$, $A\in
\mathcal{B}(\mathbb{R})$, where $\lambda_l$ is the $l$-dimensional
Lebesgue-measure restricted to $[0,1]^l$. Denote by
$F(x):=\kappa((-\infty,x])$, $x\in\mathbb{R}$ the cumulative
distribution function (c.d.f.) corresponding to $\kappa$ and set
\[
F^{-}(u):=\inf\{ q\in\mathbb{Q}: F(q)\geq u\},\ u\in (0,1).
\]
This function is measurable and it is well-known that $F^{-}(E)$ has
law $\kappa$. Now clearly
$$
(f_1(u),\ldots,f_l(u)):=\psi(F^-(u))
$$
is such that $(f_1(E),\ldots, f_l(E))$ has law $\lambda_l$ and the
$f_i$ are measurable and we get the required result, remarking that
$\lambda_l$ is the uniform law on {$[0,1]^l$}.
\end{proof}

\begin{lemma}\label{preserve}
Let $\mu(dy,dz)=\nu(y,dz)\delta(dy)$ be a probability on
$\mathbb{R}^{N_2}\times \mathbb{R}^{N_1}$ such that $\delta(dy)$ is
a probability on $\mathbb{R}^{N_2}$ and $\nu(y,dz)$ is a
probabilistic kernel. Assume that $Y$ has law $\delta(dy)$ and $E$
is independent of $Y$ and uniformly distributed on $[0,1]$. Then
there is a measurable function $G:\mathbb{R}^{N_2}\times
[0,1]\to\mathbb{R}^{N_1}$ such that $(Y, G(Y,E))$ has law
$\mu(dy,dz)$.
\end{lemma}
\begin{proof}
Just like in the previous proof, fix a Borel isomorphism
$\psi:\mathbb{R}\to\mathbb{R}^{N_1}$. Consider the measure on
$\mathbb{R}\times\mathbb{R}^{N_2}$ defined by $\tilde{\mu}(A\times
B):=\int_A\nu(y,\psi(B))\delta(dy)$,
$A\in\mathcal{B}(\mathbb{R}^{N_2})$, $B\in\mathcal{B}(\mathbb{R})$.
For $\delta$-almost every $y$, $\nu(y,\psi(\cdot))$ is a probability
measure on $\mathbb{R}$. Let $F(y,z):=\nu(y,\psi((-\infty,z])))$
denote its cumulative distribution function and define
\[
F^{-}(y,u):=\inf\{ q\in\mathbb{Q}: F(y,q)\geq u\},\ u\in (0,1),
\]
this is easily seen to be $\mathcal{B}(\mathbb{R}^{N_2}) \otimes
\mathcal{B}([0,1])$-measurable. Then, for $\delta$-almost every $y$,
$F^-(y,E)$ has law $\nu(y,\psi(\cdot))$. Hence $(Y,F^-(Y,E))$ has
law $\tilde{\mu}$. Consequently, $(Y,\psi(F^-(Y,E)))$ has law $\mu$
and we may conclude setting $G(y,u):=\psi(F^-(y,u))$. The technique
of this proof is well-known, see e.g. page 228 of \cite{bw}.
\end{proof}

The following Lemmata are used in section \ref{exaoa}. Lemma
\ref{trf} is standard and its proof is omitted.

\begin{lemma}\label{trf}
Let $X$ be a real-valued random variable with atomless law. Let
$F(x):=P(X\leq x)$ denote its cumulative distribution function. Then
$F(X)$ has uniform law on $[0,1]$.
\end{lemma}
%\begin{proof}
%The hypotheses imply that $F$ is continuous. Fix $u\in (0,1)$ and denote by $P_X$
%the law of $X$. We have
%\begin{eqnarray*}
%P(F(X)\leq u)=P_X(F^{-1}((-\infty,u]))=P_X((-\infty,x^*]),
%\end{eqnarray*}
%where $x^*=\max\{x:F(x)=u\}$ (here we use continuity of $F$). Clearly,
%$F(x^*)=u$ and hence
%\begin{eqnarray*}
%P_X((-\infty,x^*])=P(X\leq x^*)=F(x^*)=u,
%\end{eqnarray*}
%finishing the proof.
%\end{proof}

\begin{lemma}\label{kella}
Let $(X,W)$ be an $(n+m)$-dimensional random variable such that the
conditional law of $X$ w.r.t. $\sigma(W)$ is a.s. atomless. Then
there is a measurable $G:\mathbb{R}^{n+m}\to\mathbb{R}^n$ such that
$G(X,W)$ is independent of $W$ with uniform law on $[0,1]$.
\end{lemma}
\begin{proof} Let us fix a Borel-isomorphism
$\psi:\mathbb{R}^n\to\mathbb{R}$. Note that $\psi(X)$ also has an
a.s. atomless conditional law w.r.t. $\sigma(W)$.  Define (using a
regular version of the conditional law),
\begin{eqnarray*}
H(x,w):=P(\psi(X)\leq x|W=w),\ (x,w)\in\mathbb{R}\times\mathbb{R}^{m},
\end{eqnarray*}
this is $\mathcal{B}(\mathbb{R})\otimes
\mathcal{B}(\mathbb{R}^m)$-measurable (using p. 70 of
\cite{castaing-valadier} and the fact that $H$ is continuous in $x$
a.s. by hypothesis and measurable for each fixed $w$ since we took a
regular version of the conditional law). It follows by Lemma
\ref{trf} that the conditional law of $H(\psi(X),W)$ w.r.t.
$\sigma(W)$ is a.s. uniform on $[0,1]$ which means that it is
independent of $W$.
%with uniform law on $[0,1]$.
Hence we may define $G(x,w):=H(\psi(x),w)$, which is measurable
since $H$ and $\psi$ are measurable.
\end{proof}

\subsection{On a sufficient condition for Assumption \ref{as3}}\label{discussion}

%\begin{remark}\label{+++}

%For instance, if the process $\tilde{Z}_t$
%is a discretely sampled diffusion then, under mild regularity conditions,
%the joint density of $(\tilde{Z}_1,\ldots,\tilde{Z}_T)$ is continuous (even smooth)
%and bounded, see e.g. \cite{iw}.}
%\end{remark}

\subsubsection{Proof of Proposition \ref{+++}}
\label{+++appendix} Apply Corollary \ref{bank} with the choice
$k:=TN$ and $\tilde{W}_{(t-1)N+l}:=\tilde{Z}_t^l$ for $l=1,\ldots,N$
and $t=1,\ldots, T$. By the construction of Corollary \ref{bank}, {taking $Z_t^l:=W_{(t-1)N+l}$},  one has
$(\tilde{Z}_1,\ldots,\tilde{Z}_t)=g_{tN}(TN)(Z_1,\ldots,Z_t)$ hence
indeed $\mathcal{G}_t=\sigma(Z_1,\ldots,Z_t)$ for $t=1,\ldots,T$. It
is also clear that ${\Delta}S_t$ is a continuous function of
$Z_1,\ldots,Z_t$ as well.

\subsubsection{Statement and proof of Corollary \ref{bank}}
In the proof of Proposition \ref{+++}, we need Lemma \ref{end} and
its Corollary \ref{bank} below.\footnote{We thank Walter
Schachermayer for his valuable suggestions concerning Lemma
\ref{end}.}

\begin{lemma}\label{end}
Let $(\tilde{W},W)$ be a $\mathbb{R}\times\mathbb{R}^k$-valued
random variable with continuous everywhere positive density
$f(x^1,\ldots,x^{k+1})$ (with respect to the $k+1$-dimensional
Lebesgue measure) such that the function
\begin{equation}\label{tiri}
x^1\to \sup_{x^2,\ldots,x^k} f(x^1,\ldots,x^{k+1})
\end{equation}
is integrable on $\mathbb{R}$. Then there is a homeomorphism
$H:\mathbb{R}^{k+1}\to [0,1] \times \mathbb{R}^{k}$ such that
$H^i(x^1,\ldots,x^{k+1})=x^i$ for $i=2,\ldots,k+1$ and
$Z:=H^1(\tilde{W},W)$ is uniform on $[0,1]$, independent of $W$.
\end{lemma}
\begin{proof}
The conditional distribution function of $\tilde{W}$ knowing
$W=(x^2,\ldots,x^{k+1})$,
\[
F(x^1,\ldots,x^{k+1}):=
\frac{\int_{-\infty}^{x^1}f(z,x^2,\ldots,x^{k+1})dz}{\int_{-\infty}^{\infty}
f(z,x^2,\ldots,x^{k+1})dz},
\]
is continuous (due to the integrability of \eqref{tiri} and
Lebesgue's theorem). By everywhere positivity of $f$, $F$ is also
strictly increasing in $x^1$. It follows that the function
\[
H:\, (x^1,\ldots,x^{k+1})\to (F(x^1,\ldots,x^{k+1}),x^2,\ldots,x^{k+1})
\]
is a bijection and hence a homeomorhpism by Theorem 4.3 in \cite{d}.
{By Lemma \ref{trf}} the conditional law $P(H^1(\tilde{W},W)\in\cdot\,\vert
W=(x^2,\ldots,x^{k+1}))$ is uniform on $[0,1]$ for Lebesgue-almost
all $(x^2,\ldots,x^{k+1})$, which shows that $H^1(\tilde{W},W)$ is
independent of $W$ with uniform law on $[0,1]$.
\end{proof}

\begin{corollary}\label{bank}
Let $(\tilde{W}_1,\ldots,\tilde{W}_k)$ be an $\mathbb{R}^k$-valued
random variable with continuous and everywhere positive density
(w.r.t. the $k$-dimensional Lebesgue measure) such that for all
$i=1,\ldots,k$, the function
\begin{equation}\label{morrhi}
z\to \sup_{x^1,\ldots,x^{i-1}} f_i(x^1,\ldots,x^{i-1},z)
\end{equation}
is integrable on $\mathbb{R}$, where $f_i$ is the density of
$(\tilde{W}_1,\ldots, \tilde{W}_i)$ {for $i \geq 2$}. There are independent random
variables $W_1,\ldots,W_k$ and homeomorphisms
$g_l(k):\mathbb{R}^l\to\mathbb{R}^l$, $1\leq l\leq k$ such that
$(\tilde{W}_1,\ldots,\tilde{W}_l)=g_l(k)(W_1,\ldots,W_l)$.
\end{corollary}
\begin{proof} The case $k=1$ is vacuous.
Assume that the statement is true for $k\geq 1$, let us prove it for
$k+1$. We may set $g_l(k+1):=g_l(k)$, $1\leq l\leq k$, it remains to
construct $g_{k+1}(k+1)$ and $W_{k+1}$.

We wish to apply Lemma \ref{end} in this induction step with
the choice $\tilde{W}:=\tilde{W}_{k+1}$ and $W:=(\tilde{W}_1,\ldots,\tilde{W}_k)$.
In order to do this we need that
$$
z\to \sup_{x^1,\ldots,x^k} f_{k+1}(x^1,\ldots,x^k,z)
$$
is integrable, where $f_{k+1}$ is the joint density of $(\tilde{W}_1,
\ldots,\tilde{W}_{k+1})$, but this
is guaranteed by \eqref{morrhi}.

Thus Lemma \ref{end} provides a homeomorphism
$s:\mathbb{R}^{k+1}\to\mathbb{R}^{k+1}$ such that
$s^m(x^1,\ldots,x^{k+1})=x^m$, $1\leq m\leq k$ and
$W_{k+1}:=s^{k+1}(\tilde{W}_1,\ldots,\tilde{W}_{k+1})$ is
independent of $(\tilde{W}_1,\ldots,\tilde{W}_k)$ (and hence of
$(W_1,\ldots,W_k)=g_k(k)^{-1}(\tilde{W}_1,\ldots,\tilde{W}_k)$).
Define $a:\mathbb{R}^{k+1}\to\mathbb{R}^{k+1}$ by
\begin{eqnarray*}
a(x^1,\ldots,x^{k+1})& := & (g_k(k)^{-1}(x^1,\ldots,x^{k}),s^{k+1}(x^1,\ldots,x^{k+1}))\\
 & = &
s(g_k(k)^{-1}(x^1,\ldots,x^{k}),x^{k+1}),
\end{eqnarray*}
$a$ is a homeomorphism since it is the composition of two
homeomorphisms. Notice that $a(\tilde{W}_1,
\ldots,\tilde{W}_{k+1})=(W_1,\ldots,W_{k+1})$. Set
$g_{k+1}(k+1):=a^{-1}$. This finishes the proof of the induction
step and hence concludes the proof.
\end{proof}

%\bibliographystyle{plainnat}
%\bibliographystyle{apacite}
%\bibliography{biblio}

\end{document}